\documentclass[11pt, letterpaper]{article}
\usepackage[colorlinks=true,linkcolor=red,citecolor=blue]{hyperref}
\usepackage{fullpage,paralist}
\usepackage{times}
\usepackage{booktabs}
\usepackage{stmaryrd}

\usepackage{mathptmx}

\usepackage[text={6.5in,9in}]{geometry}

\usepackage{verbatim}
\usepackage{amssymb}
\usepackage{amsmath}
\usepackage{amsthm}
\usepackage{ifthen}
\usepackage{algpseudocode}
\usepackage{algorithm}
\usepackage{graphicx}
\usepackage{array}
\newcolumntype{P}[1]{>{\centering\arraybackslash}p{#1}}

\newtheorem{theorem}{Theorem}[section]

\newtheorem{lemma}{Lemma}[section]
\newtheorem{corollary}{Corollary}[section]

\def \F {{\mathbb F}}
\def \C {{\mathbb C}}
\def \R {{\mathbb R}}
\def \N {{\mathbb N}}

\newcommand{\bluetext}[1]{\textcolor{blue}{#1}}

\newcommand{\unevenmatrix}[5]{\begin{bmatrix} #1 & \begin{matrix} & \\ & \end{matrix} \\  \begin{matrix} \quad & \quad \\ & \end{matrix} & \begin{bmatrix} #2 & #3 \\ #4 & #5 \end{bmatrix} \end{bmatrix}}
\newcommand{\Mod}[1]{\ \mathrm{mod}\ #1}

\begin{document}




\title{Faster Walsh-Hadamard and Discrete Fourier Transforms\\From Matrix Non-Rigidity}

\author{Josh Alman\footnote{Columbia University, josh@cs.columbia.edu.} \and Kevin Rao\footnote{Columbia University, kevinrao99@gmail.com.}}

\maketitle
 
\begin{abstract}
    We give algorithms with lower arithmetic operation counts for both the Walsh-Hadamard Transform (WHT) and the Discrete Fourier Transform (DFT) on inputs of power-of-2 size $N$.
    
    For the WHT, our new algorithm has an operation count of $\frac{23}{24}N \log N + O(N)$. To our knowledge, this gives the first improvement on the $N \log N$ operation count of the simple, folklore Fast Walsh-Hadamard Transform algorithm.

    For the DFT, our new FFT algorithm uses $\frac{15}{4}N \log N + O(N)$ real arithmetic operations. Our leading constant $\frac{15}{4} = 3.75$ improves on the leading constant of $5$ from the Cooley-Tukey algorithm from 1965, leading constant $4$ from the split-radix algorithm of Yavne from 1968, leading constant $\frac{34}{9}=3.777\ldots$ from a modification of the split-radix algorithm by Van Buskirk from 2004, and leading constant $3.76875$ from a theoretically optimized version of Van Buskirk's algorithm by Sergeev from 2017.
    
    Our new WHT algorithm takes advantage of a recent line of work on the non-rigidity of the WHT: we decompose the WHT matrix as the sum of a low-rank matrix and a sparse matrix, and then analyze the structures of these matrices to achieve a lower operation count. Our new DFT algorithm comes from a novel reduction, showing that parts of the previous best FFT algorithms can be replaced by calls to an algorithm for the WHT. Replacing the folklore WHT algorithm with our new improved algorithm leads to our improved FFT.
\end{abstract}

\thispagestyle{empty}
\newpage
\setcounter{page}{1}

\section{Introduction}

Two of the most important and widely-used linear transforms are the Discrete Fourier Transform (DFT) and the Walsh-Hadamard Transform (WHT). In addition to their breadth of applications, these transforms can be computed quickly: They can be applied to a vector of length $N$ using only $O(N \log N)$ arithmetic operations, for instance, using the split-radix Fast Fourier Transform (FFT) algorithm~\cite{yavne1968economical} and the folklore fast Walsh–Hadamard transform, which make use of the recursive definitions of these transforms. 

Determining how much these algorithms can be sped up is an important question in both practice and theory. This is typically phrased as determining the smallest `leading constant' $c$ such that they can be computing using $(c + o(1)) N \log N$ arithmetic operations (where $\log$ denotes the base $2$ logarithm). In practice, even modest improvements to $c$ can be impactful; the current best algorithm for the DFT, which improves the leading constant from the split-radix algorithm by slightly over 5\%, is implemented today in software libraries that have been widely deployed~\cite{frigo1998fftw}\footnote{To see it in action, see lines 162 - 239 in \texttt{fftw-3.3.10/genfft/fft.ml} of the \texttt{fftw-3.3.10} package~\cite{frigo1998fftw}.}. In theory, it is popularly conjectured that an operation count of $\Omega(N \log N)$ is necessary to compute these transforms, i.e., that one cannot achieve arbitrarily small values of $c$ for either of these transforms, but this conjecture is still open. One piece of evidence for this conjecture is that there haven't been many improvements to these constants: prior to this work there have been only three FFT improvements since the original Cooley-Tukey algorithm~\cite{cooley1965algorithm} over 50 years ago, and there has never been an improvement over the folklore algorithm for the WHT. 

In this paper, we give improved algorithms, leading to smaller leading constants, for both the WHT and the DFT on inputs of power-of-two size. Our new algorithm for the WHT makes use of a recent line of work on the \emph{matrix rigidity} of the WHT, and our new algorithm for the DFT uses a novel \emph{reduction} to the WHT. Our approach for the DFT is quite different from prior Fast Fourier Transform (FFT) improvements: We focus on decreasing the number of additions and subtractions used by the  algorithms, whereas prior improvements focused on the `twiddle factor' multiplications.

\subsection{WHT Result}

Let $N$ be a power of $2$, and let $\F$ be any field. The $N \times N$ WHT matrix, $H_N$, is defined recursively by $$H_2 = \begin{bmatrix} 1 & 1 \\ 1 & -1 \end{bmatrix}, \text{ and } H_N = \begin{bmatrix} H_{N/2} & H_{N/2} \\ H_{N/2} & -H_{N/2} \end{bmatrix}$$ for $N \geq 4$. The goal of the WHT is, given as input a vector $x \in \F^{N}$, to compute the vector $y = H_N x$. We focus on $\F$ whose characteristic is not $2$, since the problem is trivial in such fields.

The folklore fast Walsh-Hadamard transform algorithm follows directly from this recursive definition: Letting $T_H(N)$ denote the number of arithmetic operations used in computing $H_N x$, we get the base case $T_H(2) = 2$, and the recurrence $T_H(N) = 2 T_H(N/2) + N$, which yield $T_H(N) = N \log N$.

In particular, it gives a leading constant of $1$, and to our knowledge, no algorithm using fewer arithmetic operations was previously known. Our first main result improves this constant from $1$ to $23/24 < 0.96$.

\begin{theorem} \label{thm:Hintro}
    For any field $\F$, given $x \in \F^N$ for $N$ a power of $2$, we can compute $H_N x$ using at most $\frac{23}{24} N \log N + \frac{13}{12}N$ field operations.
\end{theorem}

Furthermore, our algorithm only uses fairly simple field operations: $\frac{22}{24}N \log N + \frac{1}{12}N$ additions and subtractions, $\frac{1}{24}N \log N$ ``divide by 2" operations, and $N - 1$ ``multiply by a power of $2$" operations where the power of $2$ is at most $N$. When working over the reals in a typical computer architecture, these division and multiplication operations can be implemented using `bit shifts' which are substantially faster than a general addition or multiplication operation.

In Section~\ref{subsec:H4} below, we also give an alternative algorithm which uses $\frac78 N \log N + O(N)$ additions and subtractions, $\frac18 N \log N$ ``divide by 2'' operations, and $N$ ``multiply by a power of $2$" operations. Although this alternative algorithm does not use fewer total operations than the folklore Fast Walsh-Hadamard Transform, it does use fewer when  ``divide by 2'' operations are ignored, and may be even faster in certain settings.

The main idea behind our new algorithm is to take advantage of the \emph{matrix non-rigidity} of the matrix $H_N$. A matrix $M$ is called \emph{rigid} if it is impossible to change a `small' number of its entries to make it have `low' rank\footnote{See the surveys~\cite{lokam2009complexity, ramya2020recent} for a formal definition and a discussion of many applications of rigidity throughout theoretical computer science.}. Rigidity was introduced by Valiant~\cite{valiant1977graph} to prove \emph{lower bounds} on the complexity of multiplying matrices like $H_N$ by input vectors. Valiant showed that if $M$ is rigid, then one cannot multiply $M$ by an input vector using $O(N)$ arithmetic operations in an $O(\log N)$-depth arithmetic circuit. (All the algorithms discussed in this paper can be seen as $O(\log N)$-depth arithmetic circuits.) Hence, proving that $H_N$ or a similar matrix is rigid would be a first step toward proving that $\Omega(N \log N)$ operations are required to multiply it by a vector.

However, a recent line of work~\cite{alman2017probabilistic,dvir2019matrix,dvir2020fourier,alman2021kronecker,kivva2021improved,bhargava2022fast} has proved that many matrices, including the WHT~\cite{alman2017probabilistic} and DFT~\cite{dvir2020fourier} matrices, are \emph{not} rigid. A formal converse to Valiant's result is not known, and in particular, it is not immediate that this leads to improved algorithms for any of these matrices. Nonetheless, we design our faster algorithm for the WHT by making use of a rigidity upper bound from recent work of one of the authors~\cite{alman2021kronecker}, combined with new observations about the structure of changes one makes to $H_N$ to reduce its rank. 

\subsection{DFT Result}

The DFT matrix is defined over the complex numbers $\C$. Let $\omega_N := e^{i \pi / N} \in \C$ denote the $N$th root of unity, with $i = \sqrt{-1}$. The $N \times N$ DFT matrix, $F_N \in \C^{N \times N}$, is given by, for $a,b \in \{0,1,\ldots,N-1\}$,  $C[a,b] = \omega_N^{a \cdot b}$. The goal of the DFT is, given as input a vector $x \in \C^{N}$, to compute the vector $y = F_N x$. We focus here on the case where $N$ is a power of $2$.

FFT algorithms (for the DFT) have customarily been measured by the number of \emph{real} operations they perform, rather than number of complex operations, to correspond more closely with costs in real computer architectures. In other words, the problem we solve is: We are given as input $2N$ real numbers, corresponding to the real and imaginary parts of the entries of $x \in \C^N$, the goal is to output $2N$ real numbers equal to the real and imaginary parts of the entries of $y = F_N x$, and we count the number of real arithmetic operations performed. For example, we add complex numbers using $2$ real additions, and multiply two complex numbers using $6$ real operations ($4$ multiplications and $2$ additions).

Figure~\ref{fig:FFThistory} summarizes the history of the best FFT algorithms. The classic Cooley-Tukey algorithm~\cite{cooley1965algorithm} achieves a leading constant of $5$, and a few years later in 1968, Yavne~\cite{yavne1968economical} introduced the `split-radix' (SR) variant on Cooley-Tukey which improves the constant to $4$. Over thirty years later, in 2004, Van Buskirk publicly posted software \cite{buskirk2004software} which improved on the operation count of SR using a `Modified Split-Radix' (MSR) approach. By the year 2007, three different groups independently gave theoretical analyses of Van Buskirk's algorithm (or equivalent variants on it), proving its correctness and formally showing that it achieved a leading constant of $34/9 = 3.777\ldots$~\cite{lundy2007new,johnson2006modified,bernstein2007tangent}. As mentioned above, this MSR algorithm is used in widely-deployed systems today~\cite{frigo1998fftw}. More recently, in 2017, Sergeev~\cite{sergeev2017real} showed that the operation-saving technique of MSR could be applied more effectively in the limit as the `base circuit' of the algorithm gets larger and larger, improving the constant to $3.76875$.

\begin{figure}[h!]
\begin{center}
\begin{tabular}{ | c | c | }
 \hline
 \textbf{Algorithm} & \textbf{Leading constant of operation count} \\ \hline
 Cooley-Tukey~\cite{cooley1965algorithm} & $5$ \\ \hline
 Split Radix (SR)~\cite{yavne1968economical} & $4$ \\  \hline 
 Modified Split Radix (MSR)~\cite{lundy2007new,johnson2006modified,bernstein2007tangent} & $34/9 = 3.7777\ldots$  \\ \hline 

``Theoretically Optimized" Split Radix~\cite{sergeev2017real} & $3.76875$

 \\ \hline 
 This paper, Algorithm~\ref{alg:HUFFT} & $15/4 = 3.75$   \\ \hline
\end{tabular}
    \caption{History of the best FFT algorithms, for $N$ a power of $2$. For all these algorithms, a leading constant of $c$ means the algorithm uses $c N \log N + O(N)$ real operations, where the $O$ hides a modest constant.}
    \label{fig:FFThistory}
\end{center}
\end{figure}

Our second main result improves the constant to $15/4 = 3.75$:
\begin{theorem}\label{thm:Fintro}
    Given $x \in \C^N$ for $N$ a power of $2$, we can compute $F_N x$ using at most $\frac{15}{4}N \log N - \frac{223}{108}N + o(N^{0.8})$ real operations.
\end{theorem}

Furthermore, similar to Theorem~\ref{thm:Hintro}, some of the operations in Theorem~\ref{thm:Fintro} are simple multiplications or divisions by small powers of $2$ which could be implemented with fast `bit shifts': $\frac{1}{36} N \log N$ are ``divide by 2'' operations, and $N - 1$ are ``multiply by a power of $2$'' operations.

At a high level, the prior improvements beyond Cooley-Tukey focused on the number of real operations used for multiplying inputs by `twiddle factor' complex numbers throughout the algorithm. They noticed that the operation counts could be reduced by taking advantage of symmetries in the twiddle factors. For example, a key observation of SR is that, given two complex numbers $a$ and $b$, one can simultaneously compute both $a\cdot b$ and $a \cdot b^{*}$ using only 8 real operations (where $b^{*}$ denotes the complex conjugate of $b$), whereas multiplying $a$ by two arbitrary complex numbers in general is done with 12 real operations.

Our improvement focuses instead on the additions and subtractions used to combine the recursive calls at each layer of these algorithms. We show how to rearrange the order that computations are performed in, so that many of these additions and subtractions can be replaced by a reduction to the WHT. We show:
\begin{lemma} \label{lem:reduction}
    Suppose the WHT of a vector $x \in \F^N$ for $N$ a power of $2$ can be computed using $c N \log N + O(N)$ field operations. Then, the DFT of a vector $y \in \C^N$ can be performed using $\left(\frac{2}{3}c + \frac{28}{9}\right) N \log N + O(N)$ real operations.
\end{lemma}
In particular, setting $c=1$ (from the folklore fast Walsh-Hadamard transform) in Lemma~\ref{lem:reduction} recovers the leading constant $34/9$ achieved by MSR, but using our improved $c = 23/24$ from Theorem~\ref{thm:Hintro} yields the leading constant $15/4$ which we state in Theorem~\ref{thm:Fintro}. By the nature of our reduction, our new FFT algorithm is faster (in operation count or in practice) on input size $N$ whenever our algorithm from Theorem~\ref{thm:Hintro} for the WHT is faster for input size $N/8$. (Indeed, on input size $N$, our FFT algorithm involves computing many WHTs of size $N/8$ and smaller.)

Although the WHT and DFT matrices differ in critical ways, they have similar recursive structures, and it is natural to wonder whether there is a formal connection showing that an improvement to either one gives an improvement for the other. To our knowledge, our Lemma~\ref{lem:reduction} is the first such connection.

\subsection{Verification Code}

We have implemented the algorithms for Theorem~\ref{thm:Hintro} and Theorem~\ref{thm:Fintro}, which can be used to verify their correctness and operation counts. Our implementation is available at \url{code.joshalman.com/WHT-and-FFT-from-Non-Rigidity}.

\subsection{Other Related Work}

\paragraph{Barriers for the DFT}
All known algorithms for the DFT, including ours, use $\Omega(N \log N)$ arithmetic operations, but it remains an open question to prove that an algorithm using $o(N \log N)$ operations is impossible. A number of previous works have shown barrier results:  if one places some restrictions on the structure and properties of FFT algorithms, then $\Omega(N \log N)$ operations are required. For example, Morgenstern~\cite{morgenstern1973note} proved a $\tfrac{1}{2}N \log N$ lower bound for the complexity of FFT algorithms with coefficients (the fixed complex numbers that one multiplies by throughout the algorithm) of magnitude at most $1$. Works by Papadimitriou~\cite{papadimitriou1979optimality} and by Haynal and Haynal~\cite{haynal2011generating} give $\Omega(N \log N)$ lower bounds by making assumptions about the `flow graph' of the FFT algorithm, requiring, among other things, that there is a unique path in this graph from any input to any output. Pan~\cite{pan1986trade} similarly gives an $\Omega(N \log N)$ lower bound assuming the flow graph is `asynchronous'.

The Cooley-Tukey algorithm conforms to all of these assumptions, and SR conforms to all but the asynchronicity assumption. The MSR algorithm does not conform to these assumptions, and our new FFT algorithm introduces new ideas which further violate them. We use real coefficients of magnitude as large as $N$ (the ``multiply by powers of 2'' operations discussed earlier); these are larger than the coefficients used in the MSR algorithm, which were only slightly super-constant. Nonetheless, we use our coefficients in such a way that on the same inputs, the largest intermediate values one computes using our algorithm, MSR, and SR are roughly the same.\footnote{Notably, our large coefficients are all powers of $2$ which are simple to determine, and the time to compute the other `twiddle factors' used by our algorithm is nearly identical to MSR.} Our new WHT algorithm, which we use as a subroutine in our DFT algorithm, makes use of an algebraic identity which computes and combines multiple intermediate values that depend on many of the inputs. This results in a flow graph where each input has many paths to each output.

A more recent line of work by Ailon~\cite{ailon2013lower,ailon2014n,ailon2015tighter} gives an $\Omega(N \log N)$ lower bound for WHT algorithms that don't have $\Omega(n)$ different `ill-conditioned' intermediate steps. Our new WHT algorithm starts with $\Theta(n)$ steps of multiplying inputs by large powers of $2$, which are each much more ill-conditioned than is required by the barriers.

In other words, the new ideas behind our improved algorithm seemingly require further violations of the assumptions that are known to lead to $\Omega(N \log N)$ lower bounds. This suggests that new barriers are needed, but also that studying techniques that overcome these barriers more carefully may help lead to further improvements.

\paragraph{Matrix Rigidity and Linear Circuits}
Our faster WHT algorithm uses a rigidity decomposition recently given by \cite[{Lemma 4.6}]{alman2021kronecker} for the matrix $H_8$. The prior work~\cite{alman2021kronecker} used this and other rigidity decompositions to give a smaller constant-depth circuit for the WHT in a different model of computation, called the `linear circuit model'. The linear circuit model differs from our setting in how it measures complexity. While we count the number of arithmetic operations used by an algorithm, linear circuits make use of `linear gates' which may compute arbitrary linear combinations of their inputs, and they only count the total number of input wires to their gates. In particular, in that model, adding two inputs has a cost of 2 (since an addition gate would have 2 input wires), but multiplying by scalars is free. For example, replacing an addition of complex numbers with a multiplication would decrease the cost in the linear circuit model (where multiplications are free), but would not change the number of arithmetic operations that we count, and would even increase the cost when we are counting real operations (since complex multiplications are computed with 6 real operations, but additions use only 2).

Because of the differences in these models, our algorithms do not translate into improved linear circuits, and the algorithms of~\cite{alman2021kronecker} do not give improved arithmetic operation counts. Notably, in the prior work~\cite{alman2021kronecker}, only the number of entries which are changed in the rigidity decomposition seems to matter, whereas we also need to analyze the pattern of these changes. The prior work~\cite{alman2021kronecker} ultimately designs its best linear circuit by making use of a different rigidity decomposition for the matrix $H_{16}$ (which strictly improves on the circuit they design using their decomposition of $H_8$), whereas, despite some effort, we have not been able to design an improved algorithm with a decomposition of $H_{16}$ rather than the decomposition of $H_8$ that we use.

Although it is known that the DFT is not rigid \cite{dvir2020fourier}, we do not explicitly use this fact in our FFT algorithm. We only use the non-rigidity of the WHT, and then reduce the DFT to the WHT. Making use of the non-rigidity of the DFT is an exciting, open direction.

\section{Technical Overview}

We begin by giving an overview of our new algorithm for the DFT. We show how to rewrite and rearrange the computations involved in the MSR algorithm so that subcomputations can be extracted which are equivalent to the WHT. We will then give an overview of our new WHT algorithm based on a non-rigidity decomposition of the WHT matrix.

\subsection{Improving the Split-Radix Algorithm}
\label{subsec:introFFT}

The main idea behind our new FFT algorithm is to start with the MSR algorithm and perform a number of steps where its computations are rewritten or rearranged. In this overview, we will instead show how to apply these rewrites and rearrangements to the simpler SR algorithm. Since the MSR algorithm has a very similar overall structure (just with some complicated details inserted), we ultimately apply the same steps we outline here to obtain our final algorithm.

We start with the split-radix FFT algorithm (for the unfamiliar reader, see Section~\ref{subsec:SR} below for a derivation of this algorithm; for completeness, we rederive both SR and MSR in Section~\ref{sec:rederive}):

\begin{algorithm}[H]
\caption{Split-Radix FFT} \label{alg:origSR}
\begin{algorithmic}[1]

\Procedure{FFT}{$x$} $\rightarrow y_{0, 1, \dots N - 1}$ \Comment{$x \in \C^N$}
    \State $A \leftarrow FFT(x[2j]_{j = 0}^{N/2 - 1})$ \Comment{$A \in \C^{N/2}$}
    \State $B \leftarrow FFT(x[4j + 1]_{j = 0}^{N/4 - 1})$ \Comment{$B \in \C^{N/4}$}
    \State $C \leftarrow FFT(x[4j - 1]_{j = 0}^{N/4 - 1})$ \Comment{$C \in \C^{N/4}$}
    \For{$k \in [0, 1, \dots N/4 - 1]$}
        \State $y_k \leftarrow A_k + \omega_N^{k}B_k + \omega_N^{-k}C_k$ \label{sr-alg:forloop:yk}
        \State $y_{k + N/4} \leftarrow A_{k + N/4} - i \omega_N^{k}B_k + i\omega_N^{-k}C_k$ \label{sr-alg:forloop:ykn4}
        \State $y_{k + N/2} \leftarrow A_k - (\omega_N^{k}B_k + \omega_N^{-k}C_k)$
        \State $y_{k + 3N/4} \leftarrow A_{k + N/4} + i\omega_N^{k}B_k - i\omega_N^{-k}C_k$ \label{sr-alg:forloop:yk3n4}
    \EndFor
\EndProcedure

\end{algorithmic}
\end{algorithm}

We can equivalently view this algorithm in the following recursive matrix form:
$$F(x) = F_N\Pi_N^{-1} \Pi_Nx = \begin{bmatrix} I_{N/4} &  & D_N & D'_N \\  & I_{N/4} & -iD_N & iD'_N \\ I_{N/4} &  & -D_N & -D'_N \\  & I_{N/4} & iD_N & -iD'_{N} \end{bmatrix} \unevenmatrix{F_{N/2}}{F_{N/4}}{}{}{F_{N/4}}\begin{bmatrix} x[2i]_{i = 0}^{N / 2 - 1}  \\ x[4i + 1]_{i = 0}^{N / 4 - 1}  \\ x[4i - 1]_{i = 0}^{N / 4 - 1}\end{bmatrix}$$
Here, $\Pi_N$ is a permutation matrix which reorders the entries of $x$ so that the subvectors in the algorithm, to which we will recursively apply the FFT, appear contiguously for the sake of clarity.\footnote{Note that permutation matrices can be applied to an input vector without any arithmetic operations, by simply reordering the entries. In other words, the $\Pi_N$ matrix will only implicitly be implemented in any actual algorithm by the way the algorithms access their input, and thus does not add any computational complexity.} $D_{N}, D'_{N}  \in \C^{N/4 \times N/4}$ are the diagonal matrices given by $D_{N, N}[j, j] = \omega_N^j$ and $D'_{N, N}[j, j] = \omega_N^{-j}$ for $j \in \{0, \dots N/4 - 1\}$.

We now will explain our new idea that modifies this algorithm (as well as the MSR algorithm) to get a speedup. At each step, we write all changes from the previous algorithm in blue.

We begin by focusing on lines \ref{sr-alg:forloop:yk} - \ref{sr-alg:forloop:yk3n4} from the split-radix algorithm. Take, for example, line \ref{sr-alg:forloop:yk}.
$$y_k \leftarrow A_k + \omega_N^{k}B_k + \omega_N^{-k}C_k$$

The important observation here is that since $\omega_N^k$ and $\omega_N^{-k}$ are complex conjugates, they are identical except for a negated imaginary part. We will take advantage of this to rewrite the line in a convenient way; let $\alpha$ be any complex number and $\alpha^*$ be its complex conjugate. We will use the following algebraic rearrangement of the split-radix algorithm's computations (inspired by similar manipulations used in work on radix-3 FFT algorithms, e.g., \cite{suzuki1986new, dubois1978new}).

\begin{lemma}
\label{conj-trick-lemma}

Let $A, B, C, \alpha$ be complex numbers written as
\begin{align*}
    A &= a + a'i, \quad B = b + b'i, \quad C = c + c'i \\
    \alpha &= r + r'i, \quad \alpha^{*} = r - r'i
\end{align*}
for real numbers $a, a', b, b', c, c', r, r'$. Then,
$$A + (\alpha B + \alpha^{*} C) = [a + (r(b + c) + r'(c' - b'))] + [a' + (r(b' + c') + r'(b - c))]i$$
\end{lemma}

\begin{proof}
Starting with our substitution,
\begin{align*}
    A + \alpha B + \alpha^{*}C &= (a + a'i) + (r + r'i)(b + b'i) + (r - r'i)(c + c'i) \\
    &= a + rb - r'b' + rc + r'c' + a'i + rb'i + r'bi + rc'i - r'ci \\
    &= [a + (r(b + c) + r'(c' - b'))] + [a' + (r(b' + c') + r'(b - c))]i.\qedhere
\end{align*}
\end{proof}

We will analogously rewrite lines \ref{sr-alg:forloop:ykn4} through \ref{sr-alg:forloop:yk3n4} via a similar calculation:

\begin{corollary}
\label{conj-trick-corollary}
    Analogously,
    \begin{align*}
        A - (\alpha B + \alpha^{*} C) &= [a - (r(b + c) + r'(c' - b'))] + [a' - (r(b' + c') + r'(b - c))]i, \\
        A - i(\alpha B - \alpha^{*} C) &= [a + (r'(b + c) + r(b' - c'))] + [a' + (r'(b' + c') + r(c - b))]i, \\
        A + i(\alpha B - \alpha^{*} C) &= [a - (r'(b + c) + r(b' - c'))] + [a' - (r'(b' + c') + r(c - b))]i.
    \end{align*}
\end{corollary}

We use these four results and substitution to rewrite the split-radix algorithm as

\begin{algorithm}[H]
\caption{Split-Radix FFT (Rewritten)}
\begin{algorithmic}[1]

\Procedure{FFT}{$x$} $\rightarrow y_{0, 1, \dots N - 1}$ \Comment{$x \in \C^N$}
    \State $A \leftarrow FFT(x[2j]_{j = 0}^{N/2 - 1})$ \Comment{$A \in \C^{N/2}$}
    \State $B \leftarrow FFT(x[4j + 1]_{j = 0}^{N/4 - 1})$ \Comment{$B \in \C^{N/4}$}
    \State $C \leftarrow FFT(x[4j - 1]_{j = 0}^{N/4 - 1})$ \Comment{$C \in \C^{N/4}$}
    \For{$k \in [0, 1, \dots N/4 - 1]$}
        \State \bluetext{$a + a'i, z + z'i, b + b'i, c + c'i \leftarrow A_k, A_{k + N/4}, B_k, C_k$}
        \State \bluetext{$ r + r'i \leftarrow \omega_N^k$}
        \State $y_k \leftarrow \bluetext{[a + r(b + c) + r'(c' - b')] + [a' + r(b' + c') + r'(b - c)]i}$ \label{sr:before_add:yk}
        \State $y_{k + N/4} \leftarrow \bluetext{[z + r'(b + c) + r(b' - c')] + [z' + r'(b' + c') + r(c - b)]i}$
        \State $y_{k + N/2} \leftarrow \bluetext{[a - r(b + c) - r'(c' - b')] + [a' - r(b' + c') - r'(b - c)]i}$
        \State $y_{k + 3N/4} \leftarrow \bluetext{[z - r'(b + c) - r(b' - c')] + [z' - r'(b' + c') - r(c - b)]i}$ \label{sr:before_add:yk+3N/4}
    \EndFor
\EndProcedure

\end{algorithmic}
\end{algorithm}

Observe that instead of depending on $A, B, C$, the output of the $FFT$ function (i.e., the quantities calculated on lines \ref{sr:before_add:yk} through \ref{sr:before_add:yk+3N/4}) can now be thought of as depending on $a, a', b + c, b' + c', b - c, b' - c'$, which in turn depend only on $A, B + C, B - C$ (where addition and subtraction of vectors is done entry-wise). In other words, the computations in lines \ref{sr:before_add:yk} through \ref{sr:before_add:yk+3N/4} depend on $B + C = FFT(x[4j + 1]_{j = 0}^{N/4 - 1}) + FFT(x[4j - 1]_{j = 0}^{N/4 - 1})$ and $B - C = FFT(x[4j + 1]_{j = 0}^{N/4 - 1}) - FFT(x[4j - 1]_{j = 0}^{N/4 - 1})$, which by the linearity of the FFT, are equivalent to $FFT(x[4j + 1]_{j = 0}^{N/4 - 1} + x[4j - 1]_{j = 0}^{N/4 - 1})$ and $FFT(x[4j + 1]_{j = 0}^{N/4 - 1} - x[4j - 1]_{j = 0}^{N/4 - 1})$. Thus, we can replace the lines
\begin{align*}
    B &\leftarrow FFT(x[4j + 1]_{j = 0}^{N/4 - 1}) \\
    C &\leftarrow FFT(x[4j - 1]_{j = 0}^{N/4 - 1})
\end{align*}
with
\begin{align*}
    \tilde{x_B} &\leftarrow x[4j + 1]_{j = 0}^{N/4 - 1} + x[4j - 1]_{j = 0}^{N/4 - 1} \\
    \tilde{x_C} &\leftarrow x[4j + 1]_{j = 0}^{N/4 - 1} - x[4j - 1]_{j = 0}^{N/4 - 1} \\
    \tilde{B} &\leftarrow FFT(\tilde{x_B}) \\
    \tilde{C} &\leftarrow FFT(\tilde{x_C})
\end{align*}

and substitute $\tilde{b}$ for every instance of $b + c$, $\tilde{b'}$ for every instance of $b' + c'$, $\tilde{c}$ for every instance of $b - c$, and $\tilde{c'}$ for every instance of $b' - c'$, to get

\begin{algorithm}[H]
\caption{Split-Radix FFT (Intermediate modifications)}
\label{sr:alg:2.0}
\begin{algorithmic}[1]

\Procedure{FFT}{$x$} $\rightarrow y_{0, 1, \dots N - 1}$ \Comment{$x \in \C^N$}
    \State \bluetext{$\tilde{x_B} \leftarrow x[4j + 1]_{j = 0}^{N/4 - 1} + x[4j - 1]_{j = 0}^{N/4 - 1}$ \label{sr:2.0:after_add:B}}
    \State \bluetext{$\tilde{x_C} \leftarrow x[4j + 1]_{j = 0}^{N/4 - 1} - x[4j - 1]_{j = 0}^{N/4 - 1}$ \label{sr:2.0:after_add:C}}
    \State $A \leftarrow FFT(x[2j]_{j = 0}^{N/2 - 1})$ \Comment{$A \in \C^{N/2}$}
    \State \bluetext{$\tilde{B} \leftarrow FFT(\tilde{x_B})$} \Comment{$\tilde{B} \in \C^{N/4}$}
    \State \bluetext{$\tilde{C} \leftarrow FFT(\tilde{x_C})$} \Comment{$\tilde{C} \in \C^{N/4}$} 
    \For{$k \in [0, 1, \dots N/4 - 1]$}
        \State $a + a'i, z + z'i, \bluetext{\tilde{b} + \tilde{b'}i, \tilde{c} + \tilde{c'}i} \leftarrow A_k, A_{k + N/4}, \bluetext{\tilde{B}_k, \tilde{C}_k}$
        \State $r + r'i \leftarrow \omega_N^k$
        \State $y_k \leftarrow [a + (r(\bluetext{\tilde{b}}) + r'(\bluetext{-\tilde{c'}}))] + [a' + (r(\bluetext{\tilde{b'}}) + r'(\bluetext{\tilde{c}}))]i$  \label{sr:2.0:after_add:yk}
        \State $y_{k + N/4} \leftarrow [z + (r'(\bluetext{\tilde{b}}) + r(\bluetext{\tilde{c'}}))] + [z' + (r'(\bluetext{\tilde{b'}}) + r(\bluetext{-\tilde{c}}))]i$
        \State $y_{k + N/2} \leftarrow [a - (r(\bluetext{\tilde{b}}) + r'(\bluetext{-\tilde{c'}}))] + [a' - (r(\bluetext{\tilde{b'}}) + r'(\bluetext{\tilde{c}}))]i$
        \State $y_{k + 3N/4} \leftarrow [z - (r'(\bluetext{\tilde{b}}) + r(\bluetext{\tilde{c'}}))] + [z' - (r'(\bluetext{\tilde{b'}}) + r(\bluetext{-\tilde{c}}))]i$  \label{sr:2.0:after_add:yk+3N/4}
    \EndFor
\EndProcedure

\end{algorithmic}
\end{algorithm}

In this form, each layer of our recursive algorithm first does some additions and subtractions on the input, then makes recursive calls to the FFT function, then finally manipulates the results of those calls. We can now reorder the operations in the algorithm, so that it first does the additions and subtractions in lines \ref{sr:2.0:after_add:B} and \ref{sr:2.0:after_add:C} in \emph{all} of the recursive calls before performing lines \ref{sr:2.0:after_add:yk} - \ref{sr:2.0:after_add:yk+3N/4} in \emph{any} of the recursive calls. Our key insight is that if we combine all of these additions and subtractions together and do them simultaneously, there is a faster way to compute that resulting transformation by making use of our faster algorithm for the WHT.

To explain this idea more precisely, it will help to look again at the matrix form of the current algorithm. Namely, $F_N x = (F_N\Pi_N^{-1}) (\Pi_Nx)$ can be factored as the product

$$\underbrace{\begin{bmatrix} I_{N/4} &  & R_F & iR'_F \\  & I_{N/4} & R'_F & -iR_F \\  I_{N/4} &  & -R_F & -iR'_F \\  & I_{N/4} & -R'_F & iR_F \end{bmatrix}}_{TW_N}
\unevenmatrix{F_{N/2}}{F_{N/4}}{}{}{F_{N/4}}
\underbrace{\unevenmatrix{I_{N/2}}{I_{N/4}}{I_{N/4}}{I_{N/4}}{-I_{N/4}}}_{HL_N} \begin{bmatrix} x[2i]_{i = 0}^{N / 2 - 1}  \\ x[4i + 1]_{i = 0}^{N / 4 - 1}  \\ x[4i - 1]_{i = 0}^{N / 4 - 1}\end{bmatrix}$$
where $R_F$ and $R'_F$ are diagonal matrices of real numbers with $R_F[j,j] + iR'_F[j,j] = \omega_N^{j}$. \footnote{Notice that the middle two matrices in this factorization (the matrix with recursive calls and $HL$) commute. This is exactly the observation made earlier about the linearity of the FFT.}

Now we can see that in each level of our recursion we get a ``twiddle matrix" $TW_N$ to the left of our recursive calls and a ``WHT-looking matrix" $HL_N$ to the right of our recursive calls. Hence, when computing Algorithm \ref{sr:alg:2.0}, we effectively multiply $x$ on the left by a number of $HL$ matrices (corresponding to all the additions and subtractions of lines \ref{sr:2.0:after_add:B} and \ref{sr:2.0:after_add:C} in all the recursive calls) followed by a number of of $TW$ matrices (corresponding to all the manipulations of lines \ref{sr:2.0:after_add:yk} - \ref{sr:2.0:after_add:yk+3N/4} in all the recursive calls).

Thusfar, we have only rearranged computations of the split-radix algorithm, and one can verify that our current algorithm still has an identical operation count as the normal split-radix. Our improvement now comes from a new approach for simultaneously multiplying the input by all of the $HL$ matrices. Let $H'_N \in \{-1, 0, 1\}^{N \times N}$ be the linear transform corresponding to applying all the $HL$ matrices to the input $x$ of length $N$. $H'_N$ is thus recursively defined with base cases $H'_1 = \begin{bmatrix} 1 \end{bmatrix}$ and $H'_2 = \begin{bmatrix} 1 & \\ & 1 \end{bmatrix}$ and the recursion

$$H'_N = \unevenmatrix{H'_{N/2}}{H'_{N/4}}{H'_{N/4}}{H'_{N/4}}{-H'_{N/4}}.$$

In particular, from this recursive definition, we observe that $H'_N$ can be written as a permutation of a direct sum of WHT matrices (see Appendix~\ref{appendix:section:properties_hprime} for a proof), giving our reduction of the DFT to the WHT via this family of $H'$ matrices. We can thus apply our new faster algorithm for the WHT (which we describe next) in order to get an improved operation count for $H'_N$ and thus for the entire DFT\footnote{Applying the folklore WHT algorithm instead will simply give the exact same operation count as the SR FFT algorithm.}.

To summarize, after computing the transformation $H'$ on the input, we recursively call the proper twiddle matrices on the proper subsets of the input. The final result is:

\begin{algorithm}[H]
\caption{Final Split-Radix FFT (with the full ``Walsh-Hadamard Uprooting" trick)}
\label{alg:final:SR}
\begin{algorithmic}[1]

\Procedure{\bluetext{TW}}{$x$} $\rightarrow y_{0, 1, \dots N - 1}$ \Comment{$x \in \C^N$}
    \State $A \leftarrow \bluetext{TW}(x[2j]_{j = 0}^{N/2 - 1})$ \Comment{$A \in \C^{N/2}$}
    \State $\tilde{B} \leftarrow \bluetext{TW(x[4j + 1]_{j = 0}^{N/4 - 1})}$ \Comment{$B \in \C^{N/4}$} 
    \State $\tilde{C} \leftarrow \bluetext{TW(x[4j - 1]_{j = 0}^{N/4 - 1})}$ \Comment{$C \in \C^{N/4}$} 
    \For{$k \in [0, 1, \dots N/4 - 1]$}
        \State $a + a'i, z + z'i, \tilde{b} + \tilde{b'}i, \tilde{c} + \tilde{c'}i \leftarrow A_k, A_{k + N/4}, \tilde{B}_k, \tilde{C}_k$
        \State $r + r'i \leftarrow \omega_N^k$
        \State $y_k \leftarrow [a + (r(\tilde{b}) + r'(-\tilde{c'}))] + [a' + (r(\tilde{b'}) + r'(\tilde{c}))]i$
        \State $y_{k + N/4} \leftarrow [z + (r'(\tilde{b}) + r(\tilde{c'}))] + [z' + (r'(\tilde{b'}) + r(-\tilde{c}))]i$
        \State $y_{k + N/2} \leftarrow [a - (r(\tilde{b}) + r'(-\tilde{c'}))] + [a' - (r(\tilde{b'}) + r'(\tilde{c}))]i$
        \State $y_{k + 3N/4} \leftarrow [z - (r'(\tilde{b}) + r(\tilde{c'}))] + [z' - (r'(\tilde{b'}) + r(-\tilde{c}))]i$
    \EndFor
\EndProcedure
\bluetext{
\Procedure{FFT}{$x$} $\rightarrow y_{0, 1, \dots N - 1}$  \Comment{$x \in \C^N$}
    \State $y \leftarrow TW(H'(x))$
\EndProcedure
}

\end{algorithmic}
\end{algorithm}

For our overall running time analysis, we have $T_{FFT}(N) = T_{H'}(N) + T_{TW}(N)$.\footnote{Here we use the notation that $T_A(N)$ is the operation count for applying algorithm $A$ to an input vector of length $N$.} The straightforward algorithm for $H'$ would yield $T_{H'}(N) = \frac23 N \log N$, but since $H'$ is a direct sum of WHT matrices and (below) we improve the leading constant for computing the WHT by a factor of $\frac{23}{24}$, we ultimately improve the leading constant for computing $H'$ from $\frac23$ to $\frac23 \cdot \frac{23}{24} = \frac{23}{36}$. In section \ref{section:hprime-runtime}, we perform this analysis more carefully and show that $T_{H'}(N) < \frac{23}{36} N \log N + \frac{25}{12}N + o(N^{0.8})$. All that remains is to calculate $T_{TW}(N)$.

In order to implement $TW$ using as few operations as possible, we compute each of $r\tilde{b} - r'\tilde{c'}, r\tilde{b'} + r'\tilde{c}, r'\tilde{b} + r\tilde{c'}, r' \tilde{b'} - r\tilde{c}$ exactly once, and then use each of the results twice. Overall, this and the operations to combine with $a, a', z, z'$ total to 20 real operations on vectors of length $N/4$ per iteration of our loop, or $20(N/4) = 5N$ real operations total for one call of $TW(x)$ for length $N$ input $x$. This gives us the recurrence $T_{TW}(N) = 5N + T_{TW}(N/2) + 2T_{TW}(N/4) = \frac{10}{3} N \log N + O(N)$. This is the same operation count as SR achieves for the corresponding part of its calculations. Combining with $T_{H'}(N)$ gives $T_{FFT}(N) = \frac{143}{36} N \log N + O(N) \approx 3.972 N \log N + O(N)$, an improvement over the $4N \log N$ operation count of split-radix. In Section~\ref{section:msr_to_hufft} below, we apply the same ideas to the MSR algorithm to improve its operation count instead. See also Corollary~\ref{corollary:hufft_runtime_tot} for the lower order terms of the operation count.

\subsection{Faster Walsh-Hadamard Transform} \label{subsec:introWHT}

The starting point for our new algorithm for the WHT is the following decomposition of the matrix $H_8$ as the sum of a low-rank matrix and a sparse matrix, which was introduced by \cite[{Lemma 4.6}]{alman2021kronecker}:
{\fontsize{8pt}{3.6pt}
$${\text{{\normalsize $H_8$}}} = \begin{bmatrix}
1 & 1 & 1 & 1 & 1 & 1 & 1 & 1 \\
1 & -1 & 1 & -1 & 1 & -1 & 1 & -1 \\
1 & 1 & -1 & -1 & 1 & 1 & -1 & -1 \\
1 & -1 & -1 & 1 & 1 & -1 & -1 & 1 \\
1 & 1 & 1 & 1 & -1 & -1 & -1 & -1 \\
1 & -1 & 1 & -1 & -1 & 1 & -1 & 1 \\
1 & 1 & -1 & -1 & -1 & -1 & 1 & 1 \\
1 & -1 & -1 & 1 & -1 & 1 & 1 & -1 \\
\end{bmatrix} = \underbrace{\begin{bmatrix}
1 & 1 & 1 & 1 & 1 & 1 & 1 & 1 \\
1 & -1 & -1 & -1 & -1 & -1 & -1 & -1 \\
1 & -1 & -1 & -1 & -1 & -1 & -1 & -1 \\
1 & -1 & -1 & -1 & -1 & -1 & -1 & -1 \\
1 & -1 & -1 & -1 & -1 & -1 & -1 & -1 \\
1 & -1 & -1 & -1 & -1 & -1 & -1 & -1 \\
1 & -1 & -1 & -1 & -1 & -1 & -1 & -1 \\
1 & -1 & -1 & -1 & -1 & -1 & -1 & -1 \\
\end{bmatrix}}_{\text{{\normalsize low rank}}} + \underbrace{\begin{bmatrix}
0 & 0 & 0 & 0 & 0 & 0 & 0 & 0 \\
0 & 0 & 2 & 0 & 2 & 0 & 2 & 0 \\
0 & 2 & 0 & 0 & 2 & 2 & 0 & 0 \\
0 & 0 & 0 & 2 & 2 & 0 & 0 & 2 \\
0 & 2 & 2 & 2 & 0 & 0 & 0 & 0 \\
0 & 0 & 2 & 0 & 0 & 2 & 0 & 2 \\
0 & 2 & 0 & 0 & 0 & 0 & 2 & 2 \\
0 & 0 & 0 & 2 & 0 & 2 & 2 & 0 \\
\end{bmatrix}}_{\text{{\normalsize sparse}}}$$
}

Intuitively, we will use this decomposition because low rank and sparse matrices can be multiplied by a vector using few operations. Suppose we wanted to multiply $H_8$ times a length $8$ input $[a, b, c, d, e, f, g, h]$. For the low rank matrix we compute $tot = (b + c + d + e + f + g + h)$ one time and then simply compute $a + tot$ (the desired first output entry) and $a - tot$ (the desired output for all 7 other entries), for a total of 8 operations. For the sparse matrix, we can perform only 2 additions then double the result for each of the 7 nonzero rows, for a total of 21 operations. Including the 8 additions to add the results of these two matrices together, this would give an operation count of 37 for computing $H_8$.

This is a larger operation count than we are aiming for; the fast Walsh-Hadamard transform uses only $24$ operations. A key insight is that while computing these two matrices separately and adding the results is quite costly, we can reuse computations between the two. This allows us to save on computing each matrix and also on combining their results. For example, in the process of computing $tot$ we start by adding $b + c$, a value which is also used to compute the 5th row of our sparse matrix, so we can do that addition only once across the two matrices. Using observations like this, we get down to an operation count of 29.

This is still worse than the baseline of 24 operations for computing $H_8$. Our last main observation is that 7 out of these 29 operations are multiplying each of the inputs $b$ through $h$ by 2. To reduce the cost of these multiplications, we take a hint from a key idea behind the MSR algorithm for the DFT and ask: what if those inputs, which are the outputs from recursive calls, were already scaled up by a factor of 2? If instead of $[a, b, c, d, e, f, g, h]$ we received $[a, 2b, 2c, 2d, 2e, 2f, 2g, 2h]$ as input, now we can eliminate seven ``multiply by 2" operations and only divide one time on the sum $2b + 2c + 2d + 2e + 2f + 2g + 2h$ to get $tot$. In total, this would reduce the number of operations by 6, down from 29 to 23.

To achieve this, we observe that $2H(x) = H(2x)$ by the linearity of the WHT, so we can ``push down" the issue of multiplying by 2 into the recursive call. When we've reached the base case of our recursion, all of the ``multiply by 2" operations that have been pushed down finally accumulate and we multiply one time by a power of 2, thus turning many ``multiply by 2" operations into a few ``divide by 2" operations and a single ``multiply by $2^k$ for some $k$" operation.
This ultimately gives us the following algorithm based on using $H_8$ as our recursive step:

\begin{algorithm}[H]
\caption{Fast WHT from Non-rigidity of $H_8$}
\label{fast:hadamard:8x8}
\begin{algorithmic}[1]

\Procedure{$H_8$}{$x, k$} $\rightarrow y$ \Comment{$k \in \N$} \Comment{This algorithm returns $2^k H(x)$}
    \If{$N \leq 4$} Scale the inputs by $2^k$, use the folklore $N \cdot \log N$ operation WHT, \text{and \textbf{end procedure}}
    \EndIf
    \State $a \leftarrow H(x[j]_{j = 0}^{N/8 - 1}, ~k)$ \label{line:had8line8} \Comment{$a, b, c, d, e, f, g, h$ are length $N/8$ vectors.}
    \State $b \leftarrow H(x[j]_{j = N/8}^{N/4 - 1}, ~k + 1)$
    \State $c \leftarrow H(x[j]_{j = N/4}^{3N/8 - 1}, ~k + 1)$
    \State $d \leftarrow H(x[j]_{j = 3N/8}^{N/2 - 1}, ~k + 1)$
    \State $e \leftarrow H(x[j]_{j = N/2}^{5N/8 - 1}, ~k + 1)$
    \State $f \leftarrow H(x[j]_{j = 5N/8}^{3N/4 - 1}, ~k + 1)$
    \State $g \leftarrow H(x[j]_{j = 3N/4}^{7N/8 - 1}, ~k + 1)$
    \State $h \leftarrow H(x[j]_{j = 7N/8}^{N - 1}, ~k + 1)$ \label{line:had8line15}
    \State $B_1 \leftarrow b + c$\label{line:had8line16} \Comment{Addition done entry-wise}
    \State $B_2 \leftarrow d + h$
    \State $B_3 \leftarrow f + g$
    \State $tot \leftarrow B_1 + B_2 + B_3 + e$ \Comment{So $tot = b + c + d + e + f + g + h$}
    \State $tot \leftarrow tot / 2$ \Comment{Scalar division, done over all entries}
    \State $diff \leftarrow a - tot$
    \State $D \leftarrow diff + d$
    \State $E \leftarrow diff + e$
    \State $H \leftarrow diff + h$
    \State $y[j]_{j = 0}^{N/8 - 1} \leftarrow a + tot$
    \State $y[j]_{j = N/8}^{N/4 - 1} \leftarrow E + c + g$
    \State $y[j]_{j = N/4}^{3N/8 - 1} \leftarrow E + b + f$
    \State $y[j]_{j = 3N/8}^{N/2 - 1} \leftarrow E + B_2$
    \State $y[j]_{j = N/2}^{5N/8 - 1} \leftarrow D + B_1$
    \State $y[j]_{j = 5N/8}^{3N/4 - 1} \leftarrow H + c + f$
    \State $y[j]_{j = 3N/4}^{7N/8 - 1} \leftarrow H + b + g$
    \State $y[j]_{j = 7N/8}^{N - 1} \leftarrow D + B_3$\label{line:had8line32}
\EndProcedure

\end{algorithmic}
\end{algorithm}

See Section~\ref{hadamard:section} where we explain the intuition and derivation of Algorithm~\ref{fast:hadamard:8x8} in more detail. In that section, we calculate that the operation count of this algorithm is $\frac{23 }{24}N \log N + \frac{N}{24}(\log N \Mod 3) + N - 1$. The leading constant $\frac{23}{24}$ comes directly from our improvement from $24$ to $23$ operations for computing $H_8$.

\newpage

\section{Preliminaries} \label{section:prelims}

\subsection{Notation and Terminology}

For a positive integer $N$, we write $[N] := \{0,1,2,\ldots,N-1\}$. For a length-$N$ vector $x$, and an index $j \in [N]$, we write $x_j$ for the $j$th entry of $x$.

In this paper we work with the standard arithmetic circuit with fan-in 2 model, where we measure the complexity of an algorithm by the number of field operations it calls for. In the case of the WHT, since it can be computed over any field we simply count the number of field additions and field scalar multiplications. As discussed above, we measure out DFT algorithms by the number of \emph{real} arithmetic operations they perform. As is standard, the real operations of multiplying 1, -1, and 0 are free.

The following standard counts for the number of real operations used to compute complex operations will appear throughout our analyses. For two arbitrary complex numbers $a + a'i$ and $b + b'i$ (for $a,a',b,b' \in \R$), the product $(a + a'i)(b + b'i) = ab - a'b' + (ab' + a'b) i$ can be computed using the 6 real operations shown. If exactly one of $a, a', b, b'$ are equal to $1$, then the product can be performed using only 4 real operations (since two of the multiplications are now free). For a real number $b$, the product $b(a + a'i)$ can be computed with 2 real operations. Adding two arbitrary complex numbers together can also be done with 2 real operations.

In this paper we use the term ``transformation" to refer to a linear transformation on an input, ``matrix" to refer to the matrix representation of a transformation which performs the transformation when multiplied on the left of its input as a column vector, and ``algorithm" to refer to a specific procedure for computing a transformation. We refer directly to classes of matrices without subscript, for example the ``$H'$ matrices," unless we are referring to one of a specific size, in which case we write $H'_N$ to denote the $N \times N$ matrix. For example, we refer to the algorithm that generally computes the $H'$ transformation for all sizes as $H'$. This convention will also apply to the WHT, with $H$ referring generally to the WHT and $H_N$ referring to the WHT of a specific size.

For an algorithm $M$, we write $T_M(N)$ to denote the operation count of $M$ on an input of size $N$.

For a vector $x$, function $f$, and values $a$ and $b$ with $a < b$, we will write $x[f(j)]_{j = a}^b$ to notate the vector of length $b - a + 1$ whose entries are the entries of vector $x$ at indices $f(a), f(a + 1), \dots f(b)$ \textbf{in order}. We will commonly use this notation, for a length $N$ column vector $x$, in the following three ways:

$$x[2j]_{j = 0}^{N/2 - 1} = \begin{bmatrix} x_0 \\ x_2 \\ x_4 \\ \vdots \\ x_{N - 2} \end{bmatrix},\quad x[4j + 1]_{j = 0}^{N/4 - 1} = \begin{bmatrix} x_1 \\ x_5 \\ x_9 \\ \vdots \\ x_{N - 3} \end{bmatrix},\quad x[4j - 1]_{j = 0}^{N/4 - 1} = \begin{bmatrix} x_{N - 1} \\ x_3 \\ x_7 \\ \vdots \\ x_{N - 5} \end{bmatrix}$$
for the values in $x$ taken at indices which are $0 \Mod 2$, $1 \Mod 4$, and $3 \Mod 4$, respectively, and where $x[-1] = x[N - 1]$ for a zero indexed vector $x$ of length $N$.

We define the simple family of permutation matrices $\Pi_N \in \{0, 1\}^{N \times N}$ to be such that for a length $N$ vector $x$,
$$\Pi_N x = \begin{bmatrix} x[2j]_{j = 0}^{N/2 - 1} \vspace{2pt}  \\  \hline  x[4j + 1]_{j = 0}^{N/4 - 1} \vspace{2pt} \\  \hline  x[4j - 1]_{j = 0}^{N/4 - 1} \end{bmatrix}.$$
This permutation matrix $\Pi_N$ will only appear in the analysis of our DFT algorithms; it rearranges the matrices and vectors involved to be cleaner to write and read. The $\Pi_N$ matrix will only be implemented \emph{implicitly} in any actual algorithm by the way the algorithms access the indices of their input, and thus does not add any computational complexity. (In principle, any permutation matrix can be applied for free in the arithmetic circuit model, although in our case, $\Pi_N$ is a particularly simple permutation, and is the same one used in the prior SR and MSR algorithms to define a `butterfly network' structure.)

For clarity, we sometimes write block matrices such as:

$$\unevenmatrix{F_{N/2}}{F_{N/4}}{}{}{F_{N/4}}$$
to indicate an $N \times N$ matrix where the entire upper-left quadrant of the matrix is the $N/2 \times N/2$ matrix $F_{N/2}$ and the lower-right quadrant writes the $N/4 \times N/4$ matrix $F_{N/4}$ twice along the diagonal. All other entries are zero.

\subsection{The WHT}

The WHT $H_N$ on an input $x$ of length $N$ for a power of two $N$ is defined as: if we index entries of length-$N$ vectors by $j, k\in \{0,1\}^{\log N}$, then the $k$th output entry is $y_k = \sum_j (-1)^{\langle j,k \rangle} x_j$. Another way to define it is in terms of the Kronecker product $\otimes$, which is defined as

$$M \otimes N = \begin{bmatrix} m_{1, 1}N & \dots & m_{1, k}N \\ \vdots & \ddots & \vdots \\ m_{j, 1}N & \dots & m_{j, k}N \end{bmatrix} \quad \text{for} \quad M = \begin{bmatrix} m_{1, 1} & \dots & m_{1, k} \\ \vdots & \ddots & \vdots \\ m_{j, 1} & \dots & m_{j, k} \end{bmatrix}.$$ 

Using the Kronecker product, $H_N$ is recursively defined as $H_2 = \begin{bmatrix} 1 & 1 \\ 1 & -1 \end{bmatrix}$ and  $H_N = \begin{bmatrix} 1 & 1 \\ 1 & -1 \end{bmatrix} \otimes H_{N/2}$, and the goal of the WHT is to compute $y = H_N x$.

\subsection{The DFT}

The Discrete Fourier Transform of a vector $x \in \C^N$ is the vector $y \in \C^N$ given by
$y_k = \sum_{j = 0}^{N - 1} \omega_{N}^{jk}x_j$ for $k \in [N]$. Here, $\omega_{N} = e^{-2 \pi i / N}$ is the $N$th primitive root of unity (corresponding to the unit vector at an angle of $-2 \pi / N$ radians on the polar complex plane). A useful property that will be used repeatedly is that $\omega_{N}^M = \omega_{N/M}$ if $M$ divides $N$.

\section{Re-deriving Prior FFT Algorithms} \label{sec:rederive}

We begin by re-deriving the FFT algorithms from prior work which we will build off of. Readers familiar with the SR and MSR algorithms may wish to skip to the next section.

\subsection{The Split-Radix (SR) Algorithm}\label{subsec:SR} 

We begin with the SR algorithm~\cite{yavne1968economical}. Grouping terms in a convenient way (and again using the notational convenience that $x_{-1} = x_{N - 1}$), we can write the DFT as
\begin{align*}
    y_k &= \sum_{n_2 = 0}^{N/2 - 1} \omega_{N}^{2n_2k}x_{2n_2} + \sum_{n_4 = 0}^{N/4 - 1} \omega_{N}^{(4n_4 + 1)k}x_{4n_4 + 1} + \sum_{n_4 = 0}^{N/4 - 1} \omega_{N}^{(4n_4 - 1)k}x_{4n_4  - 1} \\
    &= \sum_{n_2 = 0}^{N/2 - 1} \omega_{N/2}^{n_2k}x_{2n_2} + \sum_{n_4 = 0}^{N/4 - 1} \omega_{N/4}^{n_4k}\omega_{N}^{k}x_{4n_4 + 1} + \sum_{n_4 = 0}^{N/4 - 1} \omega_{N/4}^{n_4k}\omega_{N}^{-k}x_{4n_4 - 1} \\
    &= \sum_{n_2 = 0}^{N/2 - 1} \omega_{N/2}^{n_2k}x_{2n_2} + \omega_N^k \sum_{n_4 = 0}^{N/4 - 1} \omega_{N/4}^{n_4k}x_{4n_4 + 1} + \omega_N^{-k} \sum_{n_4 = 0}^{N/4 - 1} \omega_{N/4}^{n_4k}x_{4n_4 - 1}.
\end{align*}

To turn this into the split-radix algorithm, a simple algorithm that was the state of the art for approximately 40 years, we compute that for $k \in [N / 4]$,
\begin{align*}
    y_{k + N/4} &= \sum_{n_2 = 0}^{N/2 - 1} \omega_{N/2}^{n_2(k + N/4)}x_{2n_2} + \omega_N^{k + N/4} \sum_{n_4 = 0}^{N/4 - 1} \omega_{N/4}^{n_4(k + N/4)}x_{4n_4 + 1} + \omega_N^{-(k + N/4)} \sum_{n_4 = 0}^{N/4 - 1} \omega_{N/4}^{n_4(k + N/4)}x_{4n_4 - 1} \\
    &= \sum_{n_2 = 0}^{N/2 - 1} \omega_{N/2}^{n_2(k + N/4)}x_{2n_2} + \omega_N^{k}\omega_4 \sum_{n_4 = 0}^{N/4 - 1} \omega_{N/4}^{n_4k} \omega_{N/4}^{n_4N/4}x_{4n_4 + 1} + \omega_N^{-k}\omega_4^{-1} \sum_{n_4 = 0}^{N/4 - 1} \omega_{N/4}^{n_4k} \omega_{N/4}^{n_4N/4}x_{4n_4 - 1} \\
    &= \sum_{n_2 = 0}^{N/2 - 1} \omega_{N/2}^{n_2(k + N/4)}x_{2n_2} - i \omega_N^{k} \sum_{n_4 = 0}^{N/4 - 1} \omega_{N/4}^{n_4k}x_{4n_4 + 1} + i \omega_N^{-k} \sum_{n_4 = 0}^{N/4 - 1} \omega_{N/4}^{n_4k} x_{4n_4 - 1}, \\
\end{align*}
and we can write similar expressions for $y_{k + N/2}$ and $y_{k + 3N/4}$. This is typically referred to as the ``conjugate pair'' variant of the SR algorithm, and it is the one we use throughout this paper. We now observe that our three sums contain three FFTs of size $N/2$, $N/4$, and $N/4$ applied to inputs $x[2j]_{j = 0}^{N/2 - 1}$, $x[4j + 1]_{j = 0}^{N/4 - 1}$, and $x[4j - 1]_{j = 0}^{N/4 - 1}$ (where the output of the size $N/2$ FFT is shared between $y_k$ and $y_{k + N/4}$). This gives us algorithm \ref{alg:split_radix}.

\begin{algorithm}[H] 
\caption{Split-Radix FFT}
\label{alg:split_radix}
\begin{algorithmic}[1]

\Procedure{FFT}{$x$} $\rightarrow y_{0, 1, \dots N - 1}$ \Comment{$x \in \C^N$}
    \State $A \leftarrow FFT(x[2j]_{j = 0}^{N/2 - 1})$ \Comment{$A \in \C^{N/2}$}
    \State $B \leftarrow FFT(x[4j + 1]_{j = 0}^{N/4 - 1})$ \Comment{$B \in \C^{N/4}$}
    \State $C \leftarrow FFT(x[4j - 1]_{j = 0}^{N/4 - 1})$ \Comment{$C \in \C^{N/4}$}
    \For{$k \in [0, 1, \dots N/4 - 1]$}
        \State $y_k \leftarrow A_k + \omega_N^{k}B_k + \omega_N^{-k}C_k$ \label{sr:line:basic}
        \State $y_{k + N/4} \leftarrow A_{k + N/4} - i \omega_N^{k}B_k + i\omega_N^{-k}C_k$
        \State $y_{k + N/2} \leftarrow A_k - (\omega_N^{k}B_k + \omega_N^{-k}C_k)$
        \State $y_{k + 3N/4} \leftarrow A_{k + N/4} + i\omega_N^{k}B_k - i\omega_N^{-k}C_k$
    \EndFor
\EndProcedure

\end{algorithmic}
\end{algorithm}

An equivalent definition of the DFT is as matrix multiplication $F_Nx$ where $F_N[j, k] = \omega_N^{jk}$. To picture this SR algorithm in matrix form, define $D_{N} \in \C^{N/4 \times N/4}$ to be the diagonal matrix such that $D_{N, N}[j, j] = \omega_N^j$ for $j \in \{0, \dots N/4 - 1\}$. Similarly define $D'_{N} \in \C^{N/4 \times N/4}$ where $D'_{N, N}[j, j] = \omega_N^{-j}$ for $j \in \{0, \dots N/4 - 1\}$.

This combined with the permutation matrix $\Pi_N$ lets us cleanly write the split-radix algorithm as:

$$F(x) = F_N\Pi_N^{-1} \Pi_Nx = \begin{bmatrix} I_{N/4} &  & D_N & D'_N \\  & I_{N/4} & -iD_N & iD'_N \\ I_{N/4} &  & -D_N & -D'_N \\  & I_{N/4} & iD_N & -iD'_{N} \end{bmatrix} \unevenmatrix{F_{N/2}}{F_{N/4}}{}{}{F_{N/4}}\begin{bmatrix} x[2j]_{j = 0}^{N / 2 - 1}  \\ x[4j + 1]_{j = 0}^{N / 4 - 1}  \\ x[4j - 1]_{j = 0}^{N / 4 - 1}\end{bmatrix}.$$

Each iteration of the ``for" loop can be computed in 24 entry-wise operations on vectors of length $N/4$, resulting in an operation count per call of the $FFT$ function of $24(N/4) = 6N$ operations, and an overall operation count of $T_{FFT}(N) = T_{FFT}(N/2) + 2T_{FFT}(N/4) + 6N = 4N \log N$.

\subsection{The Modified Split-Radix (MSR) Algorithm} \label{subsec:msr}

We now derive the MSR algorithm which improved the leading constant of SR from $4$ to $34/9 = 3.777\ldots$. As previously discussed, a number of different prior works presented and analyzed different, equivalent versions of this algorithm~\cite{buskirk2004software,lundy2007new,johnson2006modified,bernstein2007tangent}. We give here a presentation of the algorithm by Johnson and Frigo~\cite{johnson2006modified} which will be notationally easiest for us to work with.

The Johnson Frigo version of the MSR algorithm modifies the (conjugate pair) SR algorithm and turns it into four similar algorithms that call each other recursively. We write out the full algorithm below, but first discuss at a high level the intuition behind the algorithm and why it improves on the operation count of SR. The key observation they make is that in one line of the SR algorithm, the two complex numbers $$\omega_N^k = \cos(2\pi k / N) - i \sin(2 \pi k / N) ~\text{ and }~ \omega_N^{-k} = \cos(2\pi k / N) + i \sin(2 \pi k / N)$$ have real and complex parts of the same magnitude. Without loss of generality, suppose $\cos(2 \pi k / N) \neq 0$ (if it does equal zero, the following works on the imaginary part instead). Line \ref{sr:line:basic} of algorithm \ref{alg:split_radix} currently computes
$$A_k + (\cos(2\pi k / N) - i \sin(2 \pi k / N)) B_k + (\cos(2\pi k / N) + i \sin(2 \pi k / N))C_k\footnote{Here we just write out line \ref{sr:line:basic} of algorithm \ref{alg:split_radix} by expanding $\omega_N^k$}.$$
However, instead of this, we can compute
$$A_k + \cos(2 \pi k / N) \left[ (1 - i \tan(2 \pi k / N)) B_k + (1 + i \tan(2 \pi k / N)) C_k \right]$$
and take advantage of the fact that multiplying by a complex number with real or imaginary part 1 can be done with 4 real operations instead of 6. If we do this for all four lines in the split-radix algorithm, the total number of real operations doesn't change (we save 2 operations for each multiplication for a total of 4, but add 2 operations to multiply $\cos(2 \pi k / N)$ by both the sum and difference of the products, cancelling out our savings). That said, if somehow we can do the ``multiply by $\cos(2 \pi k / N)$'' operations fewer times, perhaps by somehow combining multiple scalar multiplications together into one scalar multiplication, we can take better advantage of the fact that multiplying by a complex number with real or imaginary part 1 costs fewer operations.

To this end, the MSR algorithm defines $k_{(4)} := k \Mod (N/4)$ and the scalar

\begin{equation*}
    s_{N, k} =
    \begin{cases}
      1, & \text{if} \ N \leq 4 \\
      s_{N/4, k_{(4)}} \cos(2 \pi k_{(4)} / N), & \text{if}\ k_{(4)} \leq N/8 \\
      s_{N/4, k_{(4)}} \sin(2 \pi k_{(4)} / N), & \text{if}\ k_{(4)} > N/8 \\
    \end{cases}
\end{equation*}

The reason both sine and cosine are used here is to avoid dividing by zero when one or the other is zero. These numbers satisfy the useful property $s_{N, k + N/4} = s_{N, k}$.

Now we have that $\frac{s_{N/4, k}}{s_{N, k}}$ is always $\frac{1}{\cos(2 \pi k_{(4)} / N)}$ or $\frac{1}{\sin(2 \pi k_{(4)} / N)}$ and $t_N^k = \omega_N^{k} \left(\frac{s_{N/4, k}}{s_{N, k}} \right)$ is always $\pm 1 \pm i\tan(2 \pi k_{(4)} / N)$ or $\pm \cot(2 \pi k_{(4)} / N) \pm i$. These will give us the desired ``cheaper'' multiplications by a complex number with real or imaginary part $\pm 1$.

Let $DFT(x)$ be the vector representing the DFT on a length $N$ input $x$. They define four functions, $F, FS, FS2, FS4$, which all output the DFT of their input, except that each entry is scaled by a different real number:

\begin{itemize}
    \item $y \leftarrow FFT(x)$, where $y_k = DFT(x)_k$
    \item $y \leftarrow FS(x)$, where $y_k = DFT(x)_k / s_{N, k}$
    \item $y \leftarrow FS2(x)$, where $y_k = DFT(x)_k / s_{2N, k}$
    \item $y \leftarrow FS4(x)$, where $y_k = DFT(x)_k / s_{4N, k}$.
\end{itemize}

And at a high level, each of these is computed as follows:

\begin{itemize}
    \item $FFT(x)$
    \begin{itemize}
        \item First make the recursive calls
        
        $A \leftarrow FFT(x[2j]_{j = 0}^{N/2 - 1})$
        
        $B \leftarrow FS(x[4j + 1]_{j = 0}^{N/4 - 1})$
        
        $C \leftarrow FS(x[4j - 1]_{j = 0}^{N/4 - 1})$
        
        \item Entries $k$ of $B$ and $C$ are scaled down by a factor of $s_{N/4, k}$ so instead of multiplying the results by $\omega_N^{k}$ as in split-radix we multiply them by $\omega_N^{k}s_{N/4, k}$ before combining them with $A$.
    \end{itemize}
    \item $FS(x)$
    \begin{itemize}
        \item First make the recursive calls
        
        $A \leftarrow FS2(x[2j]_{j = 0}^{N/2 - 1})$
        
        $B \leftarrow FS(x[4j + 1]_{j = 0}^{N/4 - 1})$
        
        $C \leftarrow FS(x[4j - 1]_{j = 0}^{N/4 - 1})$
        
        \item Entries $k$ of $B$ and $C$ are scaled by a factor of $1/s_{N/4, k}$ and entry $k$ of $A$ is scaled by a factor of $1/s_{N, k}$ (it is a length $N/2$ input and on a length $N$ input $FS2$ scales by a factor of $1/s_{2N, k}$, hence $1/s_{2(N/2), k} = 1/s_{N, k}$). So instead of multiplying the entries of $B$ and $C$ by $\omega_N^{k}$ we multiply them by $t_N^k = \omega_N^{k}\left(\frac{s_{N/4, k}}{s_{N, k}} \right)$ before combining them with $A$ so that the entire result is scaled by a factor of $1/s_{N, k}$. Furthermore, since $s_{N, k + N/4} = s_{N, k}$, for $k \in [N/4]$ $\frac{s_{N/4, k}}{s_{N, k}} = \frac{s_{N/4, k}}{s_{N, k + N/4}} = \frac{s_{N/4, k}}{s_{N, k + N/2}} = \frac{s_{N/4, k}}{s_{N, k + 3N/4}}$, allowing us to reuse computations.
        
        Since $t_N^k$ is $\pm 1 \pm i\tan(2 \pi k_{(4)} / N)$ or $\pm \cot(2 \pi k_{(4)} / N) \pm i$, either the real or imaginary part is $\pm 1$ and in total we use only $5N$ operations to combine the results of our recursive calls instead of the $6N$ operations in split-radix.
    \end{itemize}
    \item $FS2(x)$
    \begin{itemize}
        \item First make the recursive calls
        
        $A \leftarrow FS4(x[2j]_{j = 0}^{N/2 - 1})$
        
        $B \leftarrow FS(x[4j + 1]_{j = 0}^{N/4 - 1})$
        
        $C \leftarrow FS(x[4j - 1]_{j = 0}^{N/4 - 1})$
        
        \item By similar reasoning we multiply the entries of $B$ and $C$ by $t_N^k \left(\frac{s_{N, k}}{s_{2N, k}} \right)$ or $t_N^k \left(\frac{s_{N, k}}{s_{2N, k + N/k}} \right)$ where now we need to a bit more extra work than before because $s_{2N, k} \neq s_{2N, k + N/4}$ and there is not as much symmetry as in the function $FS(x)$. The extra work and savings cancel out to an operation count of 6N.
    \end{itemize}
    \item $FS4(x)$
    \begin{itemize}
        \item First make the recursive calls
        
        $A \leftarrow FS2(x[2j]_{j = 0}^{N/2 - 1})$
        
        $B \leftarrow FS(x[4j + 1]_{j = 0}^{N/4 - 1})$
        
        $C \leftarrow FS(x[4j - 1]_{j = 0}^{N/4 - 1})$
        
        \item By similar reasoning as the above, and since $s_{4N, k}, s_{4N, k + N/4}, s_{4N, k + N/2}, s_{4N, k + 3N/4}$ are all distinct, the algorithm calls for the most extra work with four additional scalar multiplications and a final operation count of 7N.
    \end{itemize}
\end{itemize}

The end result are $FFT$ and $FS2$ functions which have the same number of operations as SR, a cheaper $FS$ function that is used as frequently as possible, and a more expensive $FS4$ function whose additional operations are offset by the cheap $FS$.

These procedures are formally computed as follows.

\begin{algorithm}[H]
\caption{Modified Split-Radix FFT} \label{alg:msr}
\begin{algorithmic}[1]

\Procedure{FFT}{$x$} $\rightarrow y_{0, 1, \dots N - 1}$ \Comment{$x \in \C^N$}
    \State $A \leftarrow FFT(x[2j]_{j = 0}^{N/2 - 1})$ \Comment{$A \in \C^{N/2}$}
    \State $B \leftarrow FS(x[4j + 1]_{j = 0}^{N/4 - 1})$ \Comment{$B \in \C^{N/4}$}
    \State $C \leftarrow FS(x[4j - 1]_{j = 0}^{N/4 - 1})$ \Comment{$C \in \C^{N/4}$}
    \For{$k \in [0, 1, \dots N/4 - 1]$}
        \State $y_k \leftarrow A_k + (\omega_N^{k}s_{N/4, k}B_k + \omega_N^{-k}s_{N/4, k}C_k)$
        \State $y_{k + N/4} \leftarrow A_{k + N/4} - i( \omega_N^{k}s_{N/4, k}B_k - \omega_N^{-k}s_{N/4, k}C_k)$
        \State $y_{k + N/2} \leftarrow A_k - (\omega_N^{k}s_{N/4, k}B_k + \omega_N^{-k}s_{N/4, k}C_k)$
        \State $y_{k + 3N/4} \leftarrow A_{k + N/4} + i( \omega_N^{k}s_{N/4, k}B_k - \omega_N^{-k}s_{N/4, k}C_k)$
    \EndFor
\EndProcedure

\end{algorithmic}
\end{algorithm}

\begin{algorithm}[H]
\begin{algorithmic}[1]

\Procedure{FS}{$x$} $\rightarrow y_{0, 1, \dots N - 1}$ \Comment{$x \in \C^N$}
    \State $A \leftarrow FS2(x[2j]_{j = 0}^{N/2 - 1})$ \Comment{$A \in \C^{N/2}$}
    \State $B \leftarrow FS(x[4j + 1]_{j = 0}^{N/4 - 1})$ \Comment{$B \in \C^{N/4}$}
    \State $C \leftarrow FS(x[4j - 1]_{j = 0}^{N/4 - 1})$ \Comment{$C \in \C^{N/4}$}
    \For{$k \in [0, 1, \dots N/4 - 1]$}
        \State $y_k \leftarrow A_k + (t_N^{k}B_k + t_N^{-k}C_k)$
        \State $y_{k + N/4} \leftarrow A_{k + N/4} - i(t_N^{k}B_k - t_N^{-k}C_k)$
        \State $y_{k + N/2} \leftarrow A_k - (t_N^{k}B_k + t_N^{-k}C_k)$
        \State $y_{k + 3N/4} \leftarrow A_{k + N/4} + i(t_N^{k}B_k - t_N^{-k}C_k)$
    \EndFor
\EndProcedure

\end{algorithmic}
\end{algorithm}

\begin{algorithm}[H]
\begin{algorithmic}[1]

\Procedure{FS2}{$x$} $\rightarrow y_{0, 1, \dots N - 1}$ \Comment{$x \in \C^N$}
    \State $A \leftarrow FS4(x[2j]_{j = 0}^{N/2 - 1})$ \Comment{$A \in \C^{N/2}$}
    \State $B \leftarrow FS(x[4j + 1]_{j = 0}^{N/4 - 1})$ \Comment{$B \in \C^{N/4}$}
    \State $C \leftarrow FS(x[4j - 1]_{j = 0}^{N/4 - 1})$ \Comment{$C \in \C^{N/4}$}
    \For{$k \in [0, 1, \dots N/4 - 1]$}
        \State $y_k \leftarrow A_k + (s_{N, k} / s_{2N, k})(t_N^{k}B_k + t_N^{-k}C_k)$
        \State $y_{k + N/4} \leftarrow A_{k + N/4} - i(s_{N, k} / s_{2N, k + N/4})(t_N^{k}B_k - t_N^{-k}C_k)$
        \State $y_{k + N/2} \leftarrow A_k - (s_{N, k} / s_{2N, k})(t_N^{k}B_k + t_N^{-k}C_k)$
        \State $y_{k + 3N/4} \leftarrow A_{k + N/4} + i(s_{N, k} / s_{2N, k + N/4})(t_N^{k}B_k - t_N^{-k}C_k)$
    \EndFor
\EndProcedure

\end{algorithmic}
\end{algorithm}

\begin{algorithm}[H]
\begin{algorithmic}[1]

\Procedure{FS4}{$x$} $\rightarrow y_{0, 1, \dots N - 1}$ \Comment{$x \in \C^N$}
    \State $A \leftarrow FS2(x[2j]_{j = 0}^{N/2 - 1})$ \Comment{$A \in \C^{N/2}$}
    \State $B \leftarrow FS(x[4j + 1]_{j = 0}^{N/4 - 1})$ \Comment{$B \in \C^{N/4}$}
    \State $C \leftarrow FS(x[4j - 1]_{j = 0}^{N/4 - 1})$ \Comment{$C \in \C^{N/4}$}
    \For{$k \in [0, 1, \dots N/4 - 1]$}
        \State $y_k \leftarrow (s_{N, k} / s_{4N, k})[A_k + (t_N^{k}B_k + t_N^{-k}C_k)]$
        \State $y_{k + N/4} \leftarrow (s_{N, k} / s_{4N, k + N/4})[A_{k + N/4} - i(t_N^{k}B_k - t_N^{-k}C_k)]$
        \State $y_{k + N/2} \leftarrow (s_{N, k} / s_{4N, k + N/2})[A_k - (t_N^{k}B_k + t_N^{-k}C_k)]$
        \State $y_{k + 3N/4} \leftarrow (s_{N, k} / s_{4N, k + 3N / 4})[A_{k + N/4} + i(t_N^{k}B_k - t_N^{-k}C_k)]$
    \EndFor
\EndProcedure

\end{algorithmic}
\end{algorithm}

The operation count of the MSR algorithm is \cite{johnson2006modified}:

$$T_{FFT}(N) = \frac{34}{9} N \log N - \frac{124}{27}N - 2 \log N + \frac{10}{27}(-1)^{\log N} + 8$$

\section{Reducing the MSR Algorithm to the Walsh-Hadamard Transform} \label{section:msr_to_hufft}

We now give our new modification to MSR to reduce its operation count. The main idea is that, since the four procedures comprising MSR (Algorithm~\ref{alg:msr}) have very similar structure to SR (Algorithm~\ref{alg:origSR} in Section~\ref{subsec:introFFT}), we can perform the same modifications here to MSR as we did to SR in Section~\ref{subsec:introFFT}.

By again using the calculations from Lemma \ref{conj-trick-lemma} and Corollary \ref{conj-trick-corollary} (but with a bit of extra care with multiplying by our constant factors at the appropriate steps), followed by the same observation that the results can be seen as depending on $B + C$ and $B - C$ rather than $B$ and $C$, and then finally pulling out $H'$ as before (``uprooting" the Walsh-Hadamard recursion tree), we get four functions computing the different ``twiddle looking'' matrices corresponding to the four functions in Algorithm~\ref{alg:msr}: $TW, TWS, TWS2, TWS4$. Below, we also explicitly write out the intermediate calculations we do once and use multiple times to save operations\footnote{While we often write complex numbers in the form $(a + b) + (c + d)i$ for real values $a, b, c, d$, it is important to note that the real and imaginary parts of complex numbers are stored separately and thus we never combine $(a + b)$ with $(c + d)$.}:

\begin{algorithm}[H]
\caption{Walsh-Hadamard Uprooted FFT (WHUFFT)}
\label{alg:HUFFT}
\begin{algorithmic}[1]

\Procedure{TW}{$x$} $\rightarrow y_{0, 1, \dots N - 1}$ \Comment{$x \in \C^N$}
    \State $A \leftarrow TW(x[2j]_{j = 0}^{N/2 - 1})$ \Comment{$A \in \C^{N/2}$}
    \State $\tilde{B} \leftarrow TWS(x[4j + 1]_{j = 0}^{N/4 - 1})$ \Comment{$B \in \C^{N/4}$} 
    \State $\tilde{C} \leftarrow TWS(x[4j - 1]_{j = 0}^{N/4 - 1})$ \Comment{$C \in \C^{N/4}$} 
    \For{$k \in [0, 1, \dots N/4 - 1]$}
        \State $a + a'i, z + z'i, \tilde{b} + \tilde{b'}i, \tilde{c} + \tilde{c'}i \leftarrow A_k, A_{k + N/4}, \tilde{B}_k, \tilde{C}_k$
        \State $r + r'i \leftarrow \omega_N^ks_{N/4, k}$
        \State $D \leftarrow r(\tilde{b}) + r'(-\tilde{c'})$
        \State $E \leftarrow r(\tilde{b'}) + r'(\tilde{c})$
        \State $F \leftarrow r'(\tilde{b}) + r(\tilde{c'})$
        \State $G \leftarrow r'(\tilde{b'}) + r(-\tilde{c})$
        \State $y_k \leftarrow [a + D] + [a' + E]i$ 
        \State $y_{k + N/4} \leftarrow [z + F] + [z' + G]i$
        \State $y_{k + N/2} \leftarrow [a - D] + [a' - E]i$
        \State $y_{k + 3N/4} \leftarrow [z - F] + [z' - G]i$
    \EndFor
\EndProcedure
\end{algorithmic}
\end{algorithm}

\begin{algorithm}[H]
\begin{algorithmic}[1]

\Procedure{TWS}{$x$} $\rightarrow y_{0, 1, \dots N - 1}$ \Comment{$x \in \C^N$}
    \State $A \leftarrow TWS2(x[2j]_{j = 0}^{N/2 - 1})$ \Comment{$A \in \C^{N/2}$}
    \State $\tilde{B} \leftarrow TWS(x[4j + 1]_{j = 0}^{N/4 - 1})$ \Comment{$B \in \C^{N/4}$} 
    \State $\tilde{C} \leftarrow TWS(x[4j - 1]_{j = 0}^{N/4 - 1})$ \Comment{$C \in \C^{N/4}$} 
    \For{$k \in [0, 1, \dots N/4 - 1]$}
        \State $a + a'i, z + z'i, \tilde{b} + \tilde{b'}i, \tilde{c} + \tilde{c'}i \leftarrow A_k, A_{k + N/4}, \tilde{B}_k, \tilde{C}_k$
        \State $r + r'i \leftarrow t_N^k$
        \State $D \leftarrow r(\tilde{b}) + r'(-\tilde{c'})$
        \State $E \leftarrow r(\tilde{b'}) + r'(\tilde{c})$
        \State $F \leftarrow r'(\tilde{b}) + r(\tilde{c'})$
        \State $G \leftarrow r'(\tilde{b'}) + r(-\tilde{c})$
        \State $y_k \leftarrow [a + D] + [a' + E]i$
        \State $y_{k + N/4} \leftarrow [z + F] + [z' + G]i$
        \State $y_{k + N/2} \leftarrow [a - D] + [a' - E]i$
        \State $y_{k + 3N/4} \leftarrow [z - F] + [z' - G]i$
    \EndFor
\EndProcedure

\end{algorithmic}
\end{algorithm}

\begin{algorithm}[H]
\begin{algorithmic}[1]

\Procedure{TWS2}{$x$} $\rightarrow y_{0, 1, \dots N - 1}$ \Comment{$x \in \C^N$}
    \State $A \leftarrow TWS4(x[2j]_{j = 0}^{N/2 - 1})$ \Comment{$A \in \C^{N/2}$}
    \State $\tilde{B} \leftarrow TWS(x[4j + 1]_{j = 0}^{N/4 - 1})$ \Comment{$B \in \C^{N/4}$} 
    \State $\tilde{C} \leftarrow TWS(x[4j - 1]_{j = 0}^{N/4 - 1})$ \Comment{$C \in \C^{N/4}$} 
    \For{$k \in [0, 1, \dots N/4 - 1]$}
        \State $a + a'i, z + z'i, \tilde{b} + \tilde{b'}i, \tilde{c} + \tilde{c'}i \leftarrow A_k, A_{k + N/4}, \tilde{B}_k, \tilde{C}_k$
        \State $r + r'i \leftarrow t_N^k$
        \State $D \leftarrow r(\tilde{b}) + r'(-\tilde{c'})$
        \State $E \leftarrow r(\tilde{b'}) + r'(\tilde{c})$
        \State $F \leftarrow r'(\tilde{b}) + r(\tilde{c'})$
        \State $G \leftarrow r'(\tilde{b'}) + r(-\tilde{c})$
        \State $D' \leftarrow (\tfrac{s_{N, k}}{s_{2N, k}})(D)$
        \State $E' \leftarrow (\tfrac{s_{N, k}}{s_{2N, k}})(E)$
        \State $F' \leftarrow (\tfrac{s_{N, k}}{s_{2N, k + N/4}})(F)$
        \State $G' \leftarrow (\tfrac{s_{N, k}}{s_{2N, k + N/4}})(G)$
        \State $y_k \leftarrow [a + D'] + [a' + E']i$ 
        \State $y_{k + N/4} \leftarrow [z + F'] + [z' + G']i$
        \State $y_{k + N/2} \leftarrow [a - D'] + [a' - E']i$
        \State $y_{k + 3N/4} \leftarrow [z - F'] + [z' - G']i$
    \EndFor
\EndProcedure

\end{algorithmic}
\end{algorithm}

\begin{algorithm}[H]
\begin{algorithmic}[1]

\Procedure{TWS4}{$x$} $\rightarrow y_{0, 1, \dots N - 1}$ \Comment{$x \in \C^N$}
    \State $A \leftarrow TWS2(x[2j]_{j = 0}^{N/2 - 1})$ \Comment{$A \in \C^{N/2}$}
    \State $\tilde{B} \leftarrow TWS(x[4j + 1]_{j = 0}^{N/4 - 1})$ \Comment{$B \in \C^{N/4}$} 
    \State $\tilde{C} \leftarrow TWS(x[4j - 1]_{j = 0}^{N/4 - 1})$ \Comment{$C \in \C^{N/4}$} 
    \For{$k \in [0, 1, \dots N/4 - 1]$}
        \State $a + a'i, z + z'i, \tilde{b} + \tilde{b'}i, \tilde{c} + \tilde{c'}i \leftarrow A_k, A_{k + N/4}, \tilde{B}_k, \tilde{C}_k$
        \State $r + r'i \leftarrow t_N^k$
        \State $D \leftarrow r(\tilde{b}) + r'(-\tilde{c'})$
        \State $E \leftarrow r(\tilde{b'}) + r'(\tilde{c})$
        \State $F \leftarrow r'(\tilde{b}) + r(\tilde{c'})$
        \State $G \leftarrow r'(\tilde{b'}) + r(-\tilde{c})$
        \State $y_k \leftarrow \tfrac{s_{N, k}}{s_{4N, k}}\Big[[a + D] + [a' + E]i\Big]$ \label{line:whufft:tws4:first}
        \State $y_{k + N/4} \leftarrow \tfrac{s_{N, k}}{s_{4N, k + N/4}}\Big[[z + F] + [z' + G]i\Big]$
        \State $y_{k + N/2} \leftarrow \tfrac{s_{N, k}}{s_{4N, k + N/2}}\Big[[a - D] + [a' - E]i\Big]$
        \State $y_{k + 3N/4} \leftarrow \tfrac{s_{N, k}}{s_{4N, k + 3N/4}}\Big[[z - F] + [z' - G]i\Big]$ \label{line:whufft:tws4:fourth}
    \EndFor
\EndProcedure

\Procedure{FFT}{$x$} $\rightarrow y_{0, 1, \dots N - 1}$  \Comment{$x \in \C^N$}
    \State $y \leftarrow TW(H'(x))$ \Comment{$H'$ is computed using Algorithm~\ref{hprime:alg}.}
\EndProcedure

\end{algorithmic}
\end{algorithm}

For all four functions, the base cases are the identity function on a length 1 input, and the length-2 DFT on a length 2 input. (Indeed, since $H'_1$ and $H'_2$ are both identity matrices, these are the same as the base cases of MSR, Algorithm~\ref{alg:msr} above.)

The function $H'$ is computed using Algorithm~\ref{hprime:alg} which we discuss later in Section~\ref{section:hprime-runtime}.

\begin{theorem} \label{theorem:hufft_correctness:hufft_runtime}
    The Walsh-Hadamard Uprooted FFT (WHUFFT) algorithm (Algorithm~\ref{alg:HUFFT}) correctly computes the DFT with
    $$T_{WHUFFT}(N) = \frac{15}{4} N \log N + O(N)$$
    operations when $N$ is a power of $2$. These consist of
    \begin{itemize}
        \item $\frac{67}{18}N \log N + O(N)$ real additions and multiplications,
        \item $\frac{1}{36} N \log N$ ``divide by two" operations, and
        \item $2N - 2$ ``multiply by a power of two" operations.
    \end{itemize}
    
\end{theorem}
    
\begin{proof}
    
    The correctness follows from the fact that each method of Algorithm~\ref{alg:HUFFT} computes the same function as the corresponding method of Algorithm~\ref{alg:msr} (after $H'$ is applied to the inputs). This can be verified directly, but it follows from the fact that we got to Algorithm~\ref{alg:HUFFT} by applying the steps given in Section~\ref{subsec:introFFT} to Algorithm~\ref{alg:msr}.
    Combining this observation with Lemma \ref{hprime:alg:correctness} below proving the correctness of Algorithm \ref{hprime:alg} for computing $H'$, our Walsh-Hadamard Uprooted FFT algorithm correctly computes the DFT.

    We now count the number of real operations used by each of the functions $TW, TWS, TWS2, TWS4$ in addition to their recursive calls. In all four functions we use the common idea of computing the values $r\tilde{b} - r'\tilde{c'}, r\tilde{b'} + r'\tilde{c}, r'\tilde{b} + r\tilde{c'}, r' \tilde{b'} - r\tilde{c}$ a single time per iteration of our \textbf{for} loop, storing the intermediate values in $D, E, F, G$, respectively, and using each more than once.
    
    \begin{itemize}
        \item In each iteration of the \textbf{for} loop in $TW$, $\tilde{b}, \tilde{b'}, \tilde{c}, \tilde{c'}$ are arbitrary real numbers and $r, r'$ are real numbers derived from $\omega_N^ks_{N/4, k}$, so computing the four real numbers $D, E, F, G$ can be done using 12 operations (exactly those seen in the four expressions). We then add and subtract these four values as appropriate from $a, a', z, z'$ for a total of 20 real operations. Since we run our \textbf{for} loop for $N/4$ iterations, this gives us a total of $5N$ real operations to compute $TW$ in addition to recursive calls.
        
        \item In each iteration of the \textbf{for} loop in $TWS$,  $\tilde{b}, \tilde{b'}, \tilde{c}, \tilde{c'}$ are arbitrary real numbers but one of $r, r'$ is guaranteed to be equal to $\pm 1$ by the definition of $t_N^k$ in subsection \ref{subsec:msr}. Thus, computing $D, E, F, G$ can be done using only 8 operations since multiplying by $\pm 1$ in this setting is free. All other computations are the same as in $TW$, for a total of $16$ real operations per iteration of the \textbf{for} loop  and $4N$ real operations total to compute $TWS$ in addition to recursive calls.
        
        \item In each iteration of the \textbf{for} loop in $TWS2$, just as with $TWS$, computing $D, E, F, G$ can be done using only 8 operations.
        
        Before combining these values with $a, a', z, z'$ we first multiply:
        
        \begin{itemize}
            \item $D$ and $E$ by $\tfrac{s_{N, k}}{s_{2N, k}}$
            \item $F$ and $G$ by $\tfrac{s_{N, k}}{s_{2N, k + N/4}}$
        \end{itemize}
        
        for a total of 4 real operations more than in $TWS$. All other computations are the same as in $TWS$, for a total of $20$ real operations per iteration of the \textbf{for} loop  and $5N$ real operations total to compute $TWS2$ in addition to recursive calls.
        
        \item In each iteration of the \textbf{for} loop in $TWS4$, the computations are exactly the same as in $TWS$ except for each of lines \ref{line:whufft:tws4:first} through \ref{line:whufft:tws4:fourth} we multiply our entire result by one of four distinct real numbers before outputting the final value. Multiplying a real number by a complex number can be done with 2 real operations, so we perform 8 more real operations than in $TWS$, for a total of $24$ real operations per iteration of our \textbf{for} loop. This gives us a total of $6N$ real operations to compute $TWS4$ in addition to recursive calls.
    \end{itemize}
    
    Overall, the recurrence for the operation counts of these methods is thus:
    \begin{align*}
        T_{TW}(N) &= 5N + T_{TW}(N/2) + 2T_{TWS}(N/4) \\
        T_{TWS}(N) &= 4N + T_{TWS2}(N/2) + 2T_{TWS}(N/4) \\
        T_{TWS2}(N) &= 5N + T_{TWS4}(N/2) + 2T_{TWS}(N/4) \\
        T_{TWS4}(N) &= 6N + T_{TWS2}(N/2) + 2T_{TWS}(N/4).
    \end{align*}

In Lemma~\ref{lem:TWapproxcount} below we show that this solves to $T_{TW}(N) = \frac{28}{9} N \log N + O(N).$ (See also Lemma~\ref{lemma:hufft_twiddle_runtime} where we compute the low-order terms of the operation count.)

    Thus, our overall operation count is $T_{WHUFFT}(N) = T_{H'}(N) + T_{TW}(N) = \frac{23}{36}N \log N + \frac{28}{9} N \log N + O(N) = \frac{15}{4} N \log N + O(N)$ where the lower order terms are given in Corollary \ref{corollary:hufft_runtime_tot} in the Appendix. All operations for computing $TW(x)$ are real additions or real multiplications, and the breaking down of the $T_{WHUFFT}(N) = \frac{15}{4} N \log N + O(N)$ operations into real additions and multiplications, ``divide by two" operations, and ``multiply by a power of two" operations follows naturally from combining the corresponding counts from each part.
\end{proof}

In order to compute the $H'$ transform we will heavily rely on expressing the $H'$ matrix as its constituent WHT submatrices, so we first discuss our improved algorithms for computing the WHT transform using the non-rigidity of the WHT matrix.

\section{A Faster Algorithm for the Walsh-Hadamard Transform} 
\label{hadamard:section}

The algorithm we present in this section for the WHT works over any field $\F$, not just $\C$. We thus count the number of arithmetic operations over $\F$.

\subsection{Algorithm based on $H_4$ decomposition} \label{subsec:H4}

Before getting to our final algorithm, which is based on a non-rigidity decomposition of $H_8$, we start with a warm-up example which makes use of a similar decomposition of $H_4$~\cite{alman2021kronecker}:

$$H_4 = \begin{bmatrix}
1 & 1 & 1 & 1\\
1 & -1 & 1 & -1\\
1 & 1 & -1 & -1\\
1 & -1 & -1 & 1\\
\end{bmatrix} = \underbrace{\begin{bmatrix}
1 & 1 & 1 & 1\\
1 & -1 & -1 & -1\\
1 & -1 & -1 & -1\\
1 & -1 & -1 & -1\\
\end{bmatrix}}_{\text{{\normalsize low rank}}} + \underbrace{\begin{bmatrix}
0 & 0 & 0 & 0\\
0 & 0 & 2 & 0\\
0 & 2 & 0 & 0\\
0 & 0 & 0 & 2\\
\end{bmatrix}}_{\text{{\normalsize sparse}}}$$

This algorithm will not yet achieve an improved operation count, but it will showcase many of the final ideas. We believe it may have practical implications since so many of its arithmetic operations are actually ``divide by 2'' operations.

\begin{algorithm}[H]
\caption{Fast WHT from Non-Rigidity of $H_4$ (Final)} 
\label{fast:hadamard:4x4}
\begin{algorithmic}[1]

\Procedure{H}{$x, k$} $\rightarrow y$ \Comment{For $k \in \N$, $H(x, k)$ computes $2^k \cdot H_N(x)$.}
    \If{$N \leq 2$} Scale the inputs by $2^k$, use the folklore $N \cdot \log N$ operation WHT, \text{and \textbf{end procedure}}
    \EndIf
    \State $A \leftarrow H(x[j]_{j = 0}^{N/4 - 1}, k)$
    \State $B \leftarrow H(x[j]_{j = N/4}^{N/2 - 1}, k + 1)$ \label{line:7ofhad4}
    \State $C \leftarrow H(x[j]_{j = N/2}^{3N/4 - 1}, k + 1)$
    \State $D \leftarrow H(x[j]_{j = 3N/4}^{N - 1}, k + 1)$\label{line:9ofhad4}
    \State $E \leftarrow (B + C + D) / 2$ \Comment{Addition and division are done entry-wise}\label{line:10ofhad4}
    \State $F \leftarrow A - E$
    \State $y[j]_{j = 0}^{N/4 - 1} \leftarrow A + E$
    \State $y[j]_{j = N/4}^{N/2 - 1} \leftarrow F + C$
    \State $y[j]_{j = N/2}^{3N/4 - 1} \leftarrow F + B$
    \State $y[j]_{j = 3N/4}^{N - 1} \leftarrow F + D$\label{line:15ofhad4}
\EndProcedure

\end{algorithmic}
\end{algorithm}

\begin{theorem}
    Algorithm \ref{fast:hadamard:4x4} computes the WHT over any field $\F$ using 
    \begin{itemize}
        \item $7/8 N \log N$ field additions
        \item $1/8 N \log N$ field ``divide by two" operations
        \item $N - 1$ field ``multiply by a power of two" operations
    \end{itemize}

    when $N$ is a power of $2$
\end{theorem}

\begin{proof}

    We first prove the correctness of the following intermediate algorithm, and then demonstrate how to get our final Algorithm \ref{fast:hadamard:4x4} from it \footnote{Note that the only differences between the intermediate and final are in the base case and recursive calls for $B, C,$ and $D$; they are highlighted in blue.}:
    
    \begin{algorithm}[H]
        \caption{Fast WHT from Non-Rigidity of $H_4$ (Intermediate)} \label{alg:WHT4interm}
        \begin{algorithmic}[1]
        
        \Procedure{H}{\bluetext{$x$}} $\rightarrow y$
            \If{$N \leq 2$} \bluetext{Do not scale the inputs,} use the folklore $N \cdot \log N$ operation WHT, \text{and \textbf{end procedure}}
            \EndIf
            \State $A \leftarrow H(x[j]_{j = 0}^{N/4 - 1})$
            \State $B \leftarrow \bluetext{2H(x[j]_{j = N/4}^{N/2 - 1})}$
            \State $C \leftarrow \bluetext{2H(x[j]_{j = N/2}^{3N/4 - 1})}$
            \State $D \leftarrow \bluetext{2H(x[j]_{j = 3N/4}^{N - 1})}$
            \State $E \leftarrow (B + C + D) / 2$ \Comment{Addition done entry wise, division done over all entries}
            \State $F \leftarrow A - E$
            \State $y[j]_{j = 0}^{N/4 - 1} \leftarrow A + E$
            \State $y[j]_{j = N/4}^{N/2 - 1} \leftarrow F + C$
            \State $y[j]_{j = N/2}^{3N/4 - 1} \leftarrow F + B$
            \State $y[j]_{j = 3N/4}^{N - 1} \leftarrow F + D$
        \EndProcedure
        
        \end{algorithmic}
    \end{algorithm}
    
    First we can see that in the $N = 1$ and $N = 2$ base cases, Algorithm~\ref{alg:WHT4interm} properly computes the WHT. We now proceed by induction: supposing $H$ (from Algorithm~\ref{alg:WHT4interm}) computes the WHT for input lengths $N/4$ and $N/2$, we will show that it properly computes the WHT for input length $N$.
    
    By the recursive definition of the Hadmard transform,
    $$H_N = \begin{bmatrix} 1 & 1 & 1 & 1 \\ 1 & -1 & 1 & -1 \\ 1 & 1 & -1 & -1 \\ 1 & -1 & -1 & 1 \end{bmatrix} \otimes H_{N/4}.$$

    When we simplify the computations in Algorithm~\ref{alg:WHT4interm}, we get
    \begin{align*}
        y[j]_{j = 0}^{N/4 - 1} &\leftarrow A + E \\
        &= H(x[j]_{j = 0}^{N/4 - 1}) + \frac{2H(x[j]_{j = N/4}^{N/2 - 1}) +  2H(x[j]_{j = N/2}^{3N/4 - 1}) + 2H(x[j]_{j = 3N/4}^{N - 1})}{2} \\
        &= H(x[j]_{j = 0}^{N/4 - 1}) + H(x[j]_{j = N/4}^{N/2 - 1}) +  H(x[j]_{j = N/2}^{3N/4 - 1}) + H(x[j]_{j = 3N/4}^{N - 1})
    \end{align*}
    \begin{align*}
        y[j]_{j = 0}^{N/4 - 1} &\leftarrow F + C \\
        &= H(x[j]_{j = 0}^{N/4 - 1}) - \frac{2H(x[j]_{j = N/4}^{N/2 - 1}) +  2H(x[j]_{j = N/2}^{3N/4 - 1}) + 2H(x[j]_{j = 3N/4}^{N - 1})}{2} + 2H(x[j]_{j = N/2}^{3N/4 - 1})\\
        &= H(x[j]_{j = 0}^{N/4 - 1}) - H(x[j]_{j = N/4}^{N/2 - 1}) +  H(x[j]_{j = N/2}^{3N/4 - 1}) - H(x[j]_{j = 3N/4}^{N - 1})
    \end{align*}
    \begin{align*}
        y[j]_{j = 0}^{N/4 - 1} &\leftarrow F + B \\
        &= H(x[j]_{j = 0}^{N/4 - 1}) - \frac{2H(x[j]_{j = N/4}^{N/2 - 1}) +  2H(x[j]_{j = N/2}^{3N/4 - 1}) + 2H(x[j]_{j = 3N/4}^{N - 1})}{2} + 2H(x[j]_{j = N/4}^{N/2 - 1})\\
        &= H(x[j]_{j = 0}^{N/4 - 1}) + H(x[j]_{j = N/4}^{N/2 - 1}) -  H(x[j]_{j = N/2}^{3N/4 - 1}) - H(x[j]_{j = 3N/4}^{N - 1})
    \end{align*}
    \begin{align*}
        y[j]_{j = 0}^{N/4 - 1} &\leftarrow F + D \\
        &= H(x[j]_{j = 0}^{N/4 - 1}) - \frac{2H(x[j]_{j = N/4}^{N/2 - 1}) +  2H(x[j]_{j = N/2}^{3N/4 - 1}) + 2H(x[j]_{j = 3N/4}^{N - 1})}{2} + 2H(x[j]_{j = 3N/4}^{N - 1})\\
        &= H(x[j]_{j = 0}^{N/4 - 1}) - H(x[j]_{j = N/4}^{N/2 - 1}) -  H(x[j]_{j = N/2}^{3N/4 - 1}) + H(x[j]_{j = 3N/4}^{N - 1})
    \end{align*}
    Which exactly matches the definition. This concludes the proof of correctness of Algorithm~\ref{alg:WHT4interm}.

We next remark that Algorithm~\ref{fast:hadamard:4x4} and Algorithm~\ref{alg:WHT4interm} compute the same function (when we pick $k=0$ in Algorithm~\ref{fast:hadamard:4x4}), and so Algorithm~\ref{fast:hadamard:4x4} is also correct. Indeed, this follows from the fact that $2H(x) = H(2x)$ by the linearity of the WHT, and so picking `$k+1$' in lines \ref{line:7ofhad4} - \ref{line:9ofhad4} in Algorithm~\ref{fast:hadamard:4x4} is equivalent to multiplying those liens by $2$.
    
    It remains to count the operations used by Algorithm~\ref{fast:hadamard:4x4}.
    
    The main body of the function (lines \ref{line:10ofhad4} - \ref{line:15ofhad4}) performs $8$ arithmetic operations on vectors of length $N/4$; $7$ of these are field additions or subtractions, and $1$ is a multiplication by $1/2$. Thus, these lines overall use $7(N/4)$ field additions and subtractions, and $N/4$ scalar multiplications by $1/2$.
    
    The operation count of the base case depends on the value of $\log N \Mod 2$, since our recursion reduces the input by a factor of 4 with base cases for $N = 1, 2$. Ignoring, for a moment, the multiplications by $2^k$ in the base cases:
    
    \begin{itemize}
        \item In the $\log N \equiv 0 \Mod 2$ case, our base case is a size 1 WHT transform, which is completely free.
        \item In the $\log N \equiv 1 \Mod 2$ case, when we reach the base case level of recursion we compute a total of $N/2$ copies of the $2 \times 2$ WHT which can be each done with 2 field operations each. In this case computing our base case costs $N$ field operations.
    \end{itemize}

    Lastly, we add on $N - 1$ operations of ``multiply by a power of 2'' operations to perform all the multiplications by $2^k$ in the base cases, since every input ends up being scaled by a power of 2 other than $x_0$ (which is scaled by $2^0 = 1$).

Overall, this solves to
\begin{align*}
    T_H(N) &= 8(N/4) \left( \frac{\log N - (\log N \Mod 2)}{2} \right) + (\log N \Mod 2)N  + N - 1\\
    &= N \log N + N\left[ - (\log N \Mod 2) + (\log N \Mod 2) \right] + N - 1\\
    &= N \log N + N - 1,
\end{align*}
    where $7/8N \log N$ of the operations are field additions, $N/8 \log N$ are ``divide by two" operations, and $N - 1$ are scalar multiplications by powers of 2.
\end{proof}

Although Algorithm \ref{fast:hadamard:4x4} still has a total operation count of $> N \log N$ field operations, many of its field operations are simple ``divide by two'' or ``multiply by a power of 2'' operations which can be implemented with fast `bit shifts' rather than arbitrary field operations in many computer architectures. In other words, despite having a higher operation count, we believe it may be faster than both the folklore Fast Walsh-Hadamard Transform and our Algorithm~\ref{fast:hadamard:8x8} in certain practical settings.

\subsection{Algorithm based on $H_8$ decomposition}

We now use similar ideas to analyze Algorithm~\ref{fast:hadamard:8x8} (from Section~\ref{subsec:introWHT} above). Combining them with properties that emerge for the size-8 WHT, we get an algorithm that has a smaller operation count than the folklore WHT algorithm:

\begin{theorem} \label{fast:hadamard:8x8:thm}
    For $N$ a power of $2$, Algorithm \ref{fast:hadamard:8x8} computes the WHT over any field $\F$ using 
    \begin{itemize}
        \item $\frac{22 N \log N}{24} + \frac{N}{24}(\log N \Mod 3)$ field additions,
        \item $\frac{1}{24} N \log N$ field ``divide by two" operations, and
        \item $N - 1$ ``multiply by a power of two" operations.
    \end{itemize}

    when $N$ is a power of $2$
\end{theorem}
\begin{proof}
    We begin by proving the correctness of Algorithm \ref{fast:hadamard:8x8} by induction. We can verify the base cases $N = 1, N = 2, N = 4$ directly.
    
    For $N \geq 8$, by the recursive definition of the Hadmard transform,
    
    $$H_N = \begin{bmatrix}
    1 & 1 & 1 & 1 & 1 & 1 & 1 & 1 \\
    1 & -1 & 1 & -1 & 1 & -1 & 1 & -1 \\
    1 & 1 & -1 & -1 & 1 & 1 & -1 & -1 \\
    1 & -1 & -1 & 1 & 1 & -1 & -1 & 1 \\
    1 & 1 & 1 & 1 & -1 & -1 & -1 & -1 \\
    1 & -1 & 1 & -1 & -1 & 1 & -1 & 1 \\
    1 & 1 & -1 & -1 & -1 & -1 & 1 & 1 \\
    1 & -1 & -1 & 1 & -1 & 1 & 1 & -1 \\
    \end{bmatrix} \otimes H_{N/8}.$$
    
    Simplifying the lines of Algorithm \ref{fast:hadamard:8x8}, we see that it computes the vector $y \in \F^N$ given by
    \begin{align*}
        y[j]_{j =0}^{N/8 - 1} &\leftarrow a + tot = a + \tfrac{1}{2}(b + c + d + e + f + g + h), \\
        y[j]_{j = N/8}^{N/4 - 1} &\leftarrow E + c + g = a - tot + e + c + g = a - \tfrac{1}{2}(b - c + d - e + f - g + h), \\
        y[j]_{j = N/4}^{3N/8 - 1} &\leftarrow E + b + f = a - tot + e + b + f = a - \tfrac{1}{2}(-b + c + d - e - f + g + h), \\
        y[j]_{j = 3N/8}^{N/2 - 1} &\leftarrow E + B_2 = a - tot + e + d + h = a - \tfrac{1}{2}(b + c - d - e + f + g - h), \\
        y[j]_{j = N/2}^{5N/8 - 1} &\leftarrow D + B_1 = a - tot + d + b + c = a - \tfrac{1}{2}(-b - c - d + e + f + g + h), \\
        y[j]_{j = 5N/8}^{3N/4 - 1} &\leftarrow H + c + f = a - tot + h + c + f = a - \tfrac{1}{2}(b - c + d + e - f + g - h), \\
        y[j]_{j = 3N/4}^{7N/8 - 1} &\leftarrow H + b + g = a - tot + h + b + g = a - \tfrac{1}{2}(-b + c + d + e + f - g - h), \\
        y[j]_{j = 7N/8}^{N - 1} &\leftarrow D + B_3 = a - tot + d + f + g = a - \tfrac{1}{2}(b + c - d + e - f - g + h). \\
    \end{align*}
    Substituting the proper definitions of $b, c, d, e, f, g, h$ (as WHTs of vectors of length $N/8$) shows that it computes the desired result.

We now determine the operation count of Algorithm \ref{fast:hadamard:8x8}.

The main body of the function (lines \ref{line:had8line16} - \ref{line:had8line32} of Algorithm \ref{fast:hadamard:8x8}) is computed with $23(N/8)$ field operations, consisting of $22(N/8)$ field additions and $N/8$ scalar multiplications by $1/2$. The operation count of the base case depends on $\log N \Mod 3$, since our recursion reduces the input by a factor of 8 with three base cases for $N = 1, 2, 4$. Again ignoring the multiplications by $2^k$ for a moment:

\begin{itemize}
    \item In the $\log N \equiv 0 \Mod 3$ case, our base case is a size 1 WHT transform, which is completely free.
    \item In the $\log N \equiv 1 \Mod 3$ case, when we reach the base case level of recursion we compute a total of $N/2$ copies of the $2 \times 2$ WHT which can be each done with 2 field operations each. In this case computing our base case costs $N$ field operations.
    \item In the $\log N \equiv 2 \Mod 3$ case we compute a total of $N/4$ copies of the $4 \times 4$ WHT, which (using the standard Fast Walsh-Hadamard Transform) is done with 8 field operations (additions and subtractions), so overall computing our base case costs $2N$ field operations.
\end{itemize}

Lastly, as before, there are $N - 1$ operations of ``multiply by a power of 2" to multiply each input by a power $2^k$ other than $x_0$.

Overall, this gives us
\begin{align*}
    T_H(N) &= 23(N/8) \left( \frac{\log N - (\log N \Mod 3)}{3} \right) + (\log N \Mod 3)N  + N - 1\\
    &= \frac{23 N \log N}{24} + N\left[ - \frac{23}{24}(\log N \Mod 3) + (\log N \Mod 3) \right] + N - 1\\
    &= \frac{23 N \log N}{24} + \frac{N}{24}(\log N \Mod 3) + N - 1
\end{align*}
for the number of field operations used by algorithm \ref{fast:hadamard:8x8} to compute the WHT on an $N \times N$ sized input. All of the operations performed in the base cases are field additions and subtractions, so the operations are split into $\frac{22 N \log N}{24} + \frac{N}{24}(\log N \Mod 3)$ field additions and subtractions, $\frac{N \log N}{24}$ scalar multiplications by $1/2$, and $N - 1$ ``multiply by a power of 2" operations. In the `worst case' when $\log N \equiv 2 \Mod 3$, this is $\frac{22 N \log N}{24} + \frac{N}{12}$ field additions, $\frac{N \log N}{24}$ scalar multiplications by $1/2$, and $N - 1$ ``multiply by a power of 2" operations.

\end{proof}

\section{Complexity of $H'$}
\label{section:hprime-runtime}

Since $H'$ is used as a subroutine of Algorithm \ref{alg:HUFFT} for computing the DFT, in this section we switch back to counting real operations on complex inputs instead of counting general field operations as in the previous section.

$H'_N$ is defined only when $N$ is a power of $2$ by the base cases $H'_1 = \begin{bmatrix} 1 \end{bmatrix}$ and $H'_2 = \begin{bmatrix} 1 & \\ & 1 \end{bmatrix}$ and by induction has the recursive structure
$$H'_N = \unevenmatrix{H'_{N/2}}{H'_{N/4}}{H'_{N/4}}{H'_{N/4}}{-H'_{N/4}}.$$

This recursive definition looks similar to that of the WHT. Looking at an example such as $H'_{32}$ quickly suggests that $H'$ is (a permutation of) a direct sum of WHT matrices:

\begin{figure}[H]
    \centering
$\begin{bmatrix}\includegraphics[width=\textwidth]{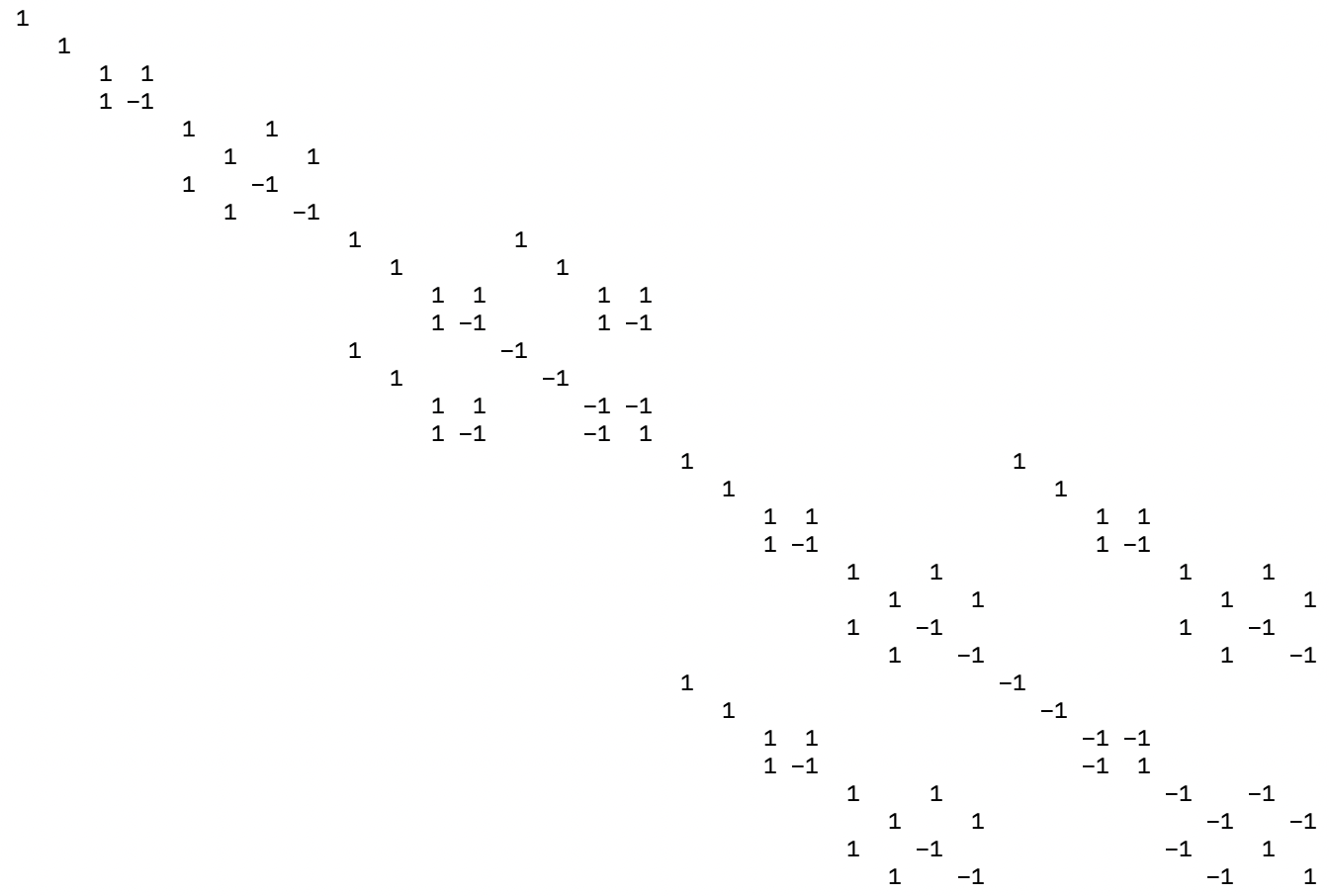}\end{bmatrix}$
    \caption{$H'_{32}$}
    \label{fig:h_prime_32}
\end{figure}

Indeed, by induction on $N$, one can see that the transformation $H'_N$ on $x$ is equivalent to performing many smaller WHTs on the different subvectors of a partition of $x$. See Appendix \ref{appendix:section:properties_hprime} below for more details. Our algorithm for quickly performing the transformation $H'_N$ then uses this obervation, and separately applies our algorithm for the WHT from the previous section to each subvector of $x$.

We first tackle the issue of partitioning the input properly. The following function, Algorithm~\ref{partition:hprime:alg}, computes an array of arrays that corresponds to the proper partition of an input of size $N$. Since it depends only on $N$, it is part of the definition of our arithmetic circuit rather than a function we compute once we are given $x$. Equivalently, we imagine precomputing it before we are given $x$ as input (at the same time as we precompute twiddle factors and other constants in our arithmetic circuit).

\begin{algorithm}[H]
\caption{Walsh-Hadamard Partition of $H'$} \label{partition:hprime:alg}
\begin{algorithmic}[1]

\Procedure{Partition}{$N$} $\rightarrow Ans$ \Comment{$N \in \N$}
    \If{$N = 1$} $Ans \leftarrow [[0]]$ \text{and \textbf{end procedure}}
    \EndIf
    \If{$N = 2$} $Ans \leftarrow [[0], [1]]$ \text{and \textbf{end procedure}}
    \EndIf
    \State $Ans \leftarrow []$ \Comment{$Ans$ is initialized as an empty array.}
    \State $P1 \leftarrow Partition(N/2)$
    \State $P2 \leftarrow Partition(N/4)$
    \For{$subset \in P1$}
        \State $Ans.append(subset)$
    \EndFor
    \For{$subset \in P2$}
        \State Instantiate temporary empty array $Temp$
        \For{$index \in subset$}
            \State $Temp.append(index + N/2)$
            \State $Temp.append(index + 3N / 4)$
        \EndFor
        \State $Ans.append(Temp)$
    \EndFor
\EndProcedure
\end{algorithmic}
\end{algorithm}

And now we can write our algorithm for computing $H'$ in terms of an algorithm $H$ that computes the WHT (assuming arrays are zero indexed).

\begin{algorithm}[H]
\caption{$H'$}
\label{hprime:alg}
\begin{algorithmic}[1]

\Procedure{$H'$}{$x$} $\rightarrow y_{0, 1, \dots N - 1}$ \Comment{$x \in \C^N$}
    \State $Partition \leftarrow \mathrm{Partition}(|x|)$ \Comment{Call algorithm \ref{partition:hprime:alg}}
    \For{$subset \in Partition$}
        \State $y[subset[j]]_{j = 0}^{|subset| - 1} \leftarrow H(x[subset[j]]_{j = 0}^{|subset| - 1})$ \label{line:call_hadmard:hprime}
    \EndFor
\EndProcedure
\end{algorithmic}
\end{algorithm}

\begin{lemma} \label{hprime:alg:correctness}
    Algorithm \ref{partition:hprime:alg} outputs a partition of the indices of $x$ where computing the transformation $H'_N$ on $x$ is equivalent to computing the WHT of each subset in the partition. Using this, Algorithm~\ref{hprime:alg} correctly computes $H'_Nx$.
\end{lemma}

\begin{proof}
    This algorithm constructs the partition in exactly the same way as we do in our proof of Theorem \ref{hprime:partition:correctness} in the Appendix \ref{appendix:section:properties_hprime}. By the definition of this partition, algorithm \ref{hprime:alg} correctly outputs $H'_Nx$.
\end{proof}

It remains to analyze Algorithm~\ref{hprime:alg}. 
Before we get into the proof, we define notation that will help formalize the relationship between the Walsh-Hadamard matrix and the $H'$ matrices. For positive integers $N_1, N_2$ which are powers of $2$, denote by $F(N_1, N_2)$ the number of copies of $H_{N_2}$ in $H'_{N_1}$ (i.e., the number of subsets $subset \in Partition(N_1)$ where $|subset| = N_2$). We have $F(1, 1) = 1$ and $F(2, 1) = 2$ as our base cases, $F(N_1, N_2) = 0$ for $N_2 < 1$ and for $N_2 > N_1$ \footnote{The $N_2 < 1$ and $N_2 > N_1$ degenerate cases will come into play in a later proof.}, and the recurrence $F(N_1, N_2) = F(N_1 / 2, N_2) + F(N_1 / 4, N_2 / 2)$, which follows from the recursive definition of $H'$, and which we prove in Appendix \ref{appendix:section:properties_hprime}.

Before calculating the operation count of Algorithm~\ref{hprime:alg}, we begin with three useful lemmas about $F(N_1, N_2)$.

\begin{lemma} \label{lemma:sum_of_subsets_and_F}
    $$\sum_{j = 0}^{\log N} F(N, 2^j) 2^j = N$$
\end{lemma}

\begin{proof}
    
    Since $F(N, 2^j)$ is the number of size $2^j$ subsets in our partition of a length $N$ input, $\sum_{j = 0}^{\log N} F(N, 2^j) 2^j$ is just the total size of all subsets in our partition which by definition is $N$.
\end{proof}

    \begin{lemma} \label{lemma:sum-of-mod}
        $$\sum_{j = 0}^{\log N}F(N, 2^j) \frac{2^j}{12}(j \Mod 3) < N/12 + o(N^{0.8})$$
    \end{lemma}
    
    In light of Lemma~\ref{lemma:sum_of_subsets_and_F}, this Lemma~\ref{lemma:sum-of-mod} is not too surprising. The proof requires carefully bounding how much `error' can accrue from the ($j \Mod 3$) term, which we do by diagonalizing the matrix for the recursion defining this sum. Since the mathematical details are not particularly enlightening, we defer the full proof to Appendix~\ref{appendix:proof:sum-of-mod}.

    \begin{lemma} \label{hprime:complexity_vs_hadamard:lemma}
        $$\sum_{j = 0}^{\log N}F(N, 2^j) 2^jj = \frac{1}{3}N \log N + \frac{2}{9}(-1)^{\log N} - \frac{2}{9}N$$
    \end{lemma}
    \begin{proof}
        Let $f(N) = \sum_{j = 0}^{\log N}F(N, 2^j) 2^jj$. Then,
        \begin{align*}
            f(N) &= \sum_{j = 0}^{\log N}F(N, 2^j) 2^jj \\
            &= \sum_{j = 0}^{\log N} \left[ F(N/2, 2^j) + F(N/4, 2^{j - 1}) \right] 2^jj \\
            &= \sum_{j = 0}^{\log N} F(N/2, 2^j)2^jj + \sum_{j = 0}^{\log N} F(N/4, 2^{j - 1})2^jj \\
            &= \sum_{j = 0}^{\log N} F(N/2, 2^j)2^jj + 2\sum_{j = 0}^{\log N} F(N/4, 2^{j})2^{j}(j + 1) \\
            &= \sum_{j = 0}^{\log N} F(N/2, 2^j)2^jj + 2\left[\sum_{j = 0}^{\log N} F(N/4, 2^{j})2^{j}j + \sum_{j = 0}^{\log N} F(N/4, 2^{j})2^{j} \right] \\
            &= \sum_{j = 0}^{\log N} F(N/2, 2^j)2^jj + 2\sum_{j = 0}^{\log N} F(N/4, 2^{j})2^{j}j + N/2 \\
             &= \sum_{j = 0}^{\log N - 1} F(N/2, 2^j)2^jj + 2\sum_{j = 0}^{\log N - 2} F(N/4, 2^{j})2^{j}j + N/2, \\
            f(N) &= f(N/2) + 2f(N/4) + N/2. \\
        \end{align*}
    
        With the base cases $f(1) = 0$ and $f(2) = 0$, this recursion solves to 
        
        $$f(N) = \sum_{j = 0}^{\log N}F(N, 2^j) 2^jj = \frac{1}{3}N \log N + \frac{2}{9}(-1)^{\log N} - \frac{2}{9}N < \frac{1}{3} N \log N.$$
    \end{proof}

We are now ready to calculate the operation count:

\begin{theorem} \label{hprime:complexity:thm}
If Algorithm~\ref{fast:hadamard:8x8} is used to compute $H$ on line~\ref{line:call_hadmard:hprime} of Algorithm \ref{hprime:alg}, then Algorithm \ref{hprime:alg} computes $H'_Nx$ in at most $\frac{23}{36}N \log N + \frac{25}{12}N + o(N^{0.8})$ real operations, consisting of
    
    \begin{itemize}
        \item $\frac{22}{36} N \log N + \frac{N}{12} + o(N^{0.8})$ real additions
        \item $\frac{1}{36} N \log N$ ``divide by two" operations
        \item $2N - 2$ ``multiply by a power of two" operations
    \end{itemize}
\end{theorem}

\begin{proof}

    We now compute the absolute number of real operations involved in computing the transformation $H'$. This relationship between $H'$ and the WHT leads us to the following conclusion about the operation count of the $H'$:
    $$T_{H'}(N) = \sum_{j = 0}^{\log N}F(N, 2^j) T_H(2^j),$$
    where $T_H(N)$ is the real operation count of the WHT on a complex valued input of length $N$. 
    
    First, see Theorem \ref{fast:hadamard:8x8:thm} for a result that $T_H(N) = \frac{46 N \log N}{24} + \frac{N}{12}(\log N \Mod 3) + 2N - 2$ real operations over the field of complex numbers (a complex addition or scalar multiplication is done with two real operations).
    
    We use our Lemmas about $F(N_1, N_2)$ to compute:
    \begin{align*}
        T_{H'}(N) &= \sum_{j = 0}^{\log N}F(N, 2^j) T_H(2^j) \\
        &= \sum_{j = 0}^{\log N}F(N, 2^j) \left[ \frac{46 (2^j)(j)}{24} + \frac{2^j}{12}(j \Mod 3) + 2(2^j) - 2\right] \\
        &= \sum_{j = 0}^{\log N}F(N, 2^j) \frac{46 (2^j)(j)}{24} + \sum_{j = 0}^{\log N}F(N, 2^j) \frac{2^j}{12}(j \Mod 3)  + \sum_{j = 0}^{\log N}F(N, 2^j) [2(2^j) - 2] \\
        &= \frac{46}{24} \left[\sum_{j = 0}^{\log N}F(N, 2^j) 2^jj \right] + \sum_{j = 0}^{\log N}F(N, 2^j) \frac{2^j}{12}(j \Mod 3) + 2N - 2\left(\sum_{j = 0}^{\log N}F(N, 2^j) \right)
    \end{align*}

    Applying Lemma~\ref{hprime:complexity_vs_hadamard:lemma},
    \begin{align*}
        &= \frac{46}{24} \left[\sum_{j = 0}^{\log N}F(N, 2^j) 2^jj \right] + \sum_{j = 0}^{\log N}F(N, 2^j) \frac{2^j}{12}(j \Mod 3) + 2N - 2\left(\sum_{j = 0}^{\log N}F(N, 2^j) \right) \\
        &< \frac{23}{36} N \log N + \sum_{j = 0}^{\log N}F(N, 2^j) \frac{2^j}{12}(j \Mod 3)  + 2N - 2\left(\sum_{j = 0}^{\log N}F(N, 2^j) \right)
    \end{align*}
    
    To compute the latter part of our operation count, we use Lemma~\ref{lemma:sum-of-mod}, along with the simple observation that 
        $$2N - 2\left(\sum_{j = 0}^{\log N}F(N, 2^j) \right) < 2N$$
    to conlcude
    $$T_{H'}(N) < \frac{23}{36}N \log N + \frac{25}{12}N + o(N^{0.8}).$$
\end{proof}

\section{Asymptotic Complexity of FFT in Relation to Walsh-Hadamard-Transform}

By separating out $H'$ in our FFT algorithm and identifying the relationship between $H'$ and the WHT as in the previous section, we can directly relate the complexity of computing the DFT to the complexity of computing the WHT, from which Lemma~\ref{lem:reduction} follows.

We conclude with a special case of Lemma~\ref{lem:reduction}, to determine what would happen for the DFT if one were to give an $o(N \log N)$ operation count algorithm for the WHT.

\begin{theorem}
    If $T_H(N) = o(N \log N)$, then $T_{H'}(N) = o(N \log N)$
\end{theorem}

\begin{proof}
    Recall the definition
    
    $$T_{H'}(N) = \sum_{j = 0}^{\log N}F(N, 2^j) T_H(2^j).$$
    
    Since we assume $T_H(N) = o(N \log N)$, let $T_H(N) = Nf(N)$ where $f(N) = o(\log N)$. Substituting gives
    \begin{align*}
        T_{H'}(N) &= \sum_{j = 0}^{\log N}F(N, 2^j) T_H(2^j) \\
        &= \sum_{j = 0}^{\log N}F(N, 2^j) 2^j f(2^j) \\
        &< \sum_{j = 0}^{\log N}F(N, 2^j) 2^j f(N) \\
        &= f(N) \sum_{j = 0}^{\log N}F(N, 2^j) 2^j \\
        &= Nf(N) \\
        &= o(N \log N).
    \end{align*}
\end{proof}

Combining this with our results from before, this means that if $T_H(N) = o(N \log N)$, then Algorithm~\ref{alg:HUFFT} becomes an FFT algorithm with $T_{FFT}(N) = \frac{28}{9} N \log N + o(N \log N)$.

\section{Acknowledgements}

We would like to thank Nir Ailon, Chi-Ning Chou, Sandeep Silwal, and anonymous reviewers for helpful comments on an earlier draft and Igor Sergeev for answering our questions about his algorithm in \cite{sergeev2017real}. This research was supported in part by NSF Grant CCF-2238221 and a grant from the
Simons Foundation (Grant Number 825870 JA).

\newpage

\bibliographystyle{alpha}
\bibliography{refs}

\newcommand{\etalchar}[1]{$^{#1}$}
\begin{thebibliography}{BGKM22}

\bibitem[Ail13]{ailon2013lower}
Nir Ailon.
\newblock A lower bound for fourier transform computation in a linear model
  over 2x2 unitary gates using matrix entropy.
\newblock {\em arXiv preprint arXiv:1305.4745}, 2013.

\bibitem[Ail14]{ailon2014n}
Nir Ailon.
\newblock An n$\backslash$log n lower bound for fourier transform computation
  in the well conditioned model.
\newblock {\em arXiv preprint arXiv:1403.1307}, 2014.

\bibitem[Ail15]{ailon2015tighter}
Nir Ailon.
\newblock Tighter fourier transform lower bounds.
\newblock In {\em Automata, Languages, and Programming: 42nd International
  Colloquium, ICALP 2015, Kyoto, Japan, July 6-10, 2015, Proceedings, Part I},
  pages 14--25. Springer, 2015.

\bibitem[Alm21]{alman2021kronecker}
Josh Alman.
\newblock Kronecker products, low-depth circuits, and matrix rigidity.
\newblock In {\em Proceedings of the 53rd Annual ACM SIGACT Symposium on Theory
  of Computing}, pages 772--785, 2021.

\bibitem[AW17]{alman2017probabilistic}
Josh Alman and Ryan Williams.
\newblock Probabilistic rank and matrix rigidity.
\newblock In {\em Proceedings of the 49th Annual ACM SIGACT Symposium on Theory
  of Computing}, pages 641--652, 2017.

\bibitem[Ber07]{bernstein2007tangent}
Daniel~J Bernstein.
\newblock The tangent fft.
\newblock In {\em International Symposium on Applied Algebra, Algebraic
  Algorithms, and Error-Correcting Codes}, pages 291--300. Springer, 2007.

\bibitem[BGKM22]{bhargava2022fast}
Vishwas Bhargava, Sumanta Ghosh, Mrinal Kumar, and Chandra~Kanta Mohapatra.
\newblock Fast, algebraic multivariate multipoint evaluation in small
  characteristic and applications.
\newblock In {\em Proceedings of the 54th Annual ACM SIGACT Symposium on Theory
  of Computing}, pages 403--415, 2022.

\bibitem[CT65]{cooley1965algorithm}
James~W Cooley and John~W Tukey.
\newblock An algorithm for the machine calculation of complex fourier series.
\newblock {\em Mathematics of computation}, 19(90):297--301, 1965.

\bibitem[DE19]{dvir2019matrix}
Zeev Dvir and Benjamin~L Edelman.
\newblock Matrix rigidity and the croot-lev-pach lemma.
\newblock {\em Theory Of Computing}, 15(8):1--7, 2019.

\bibitem[DL20]{dvir2020fourier}
Zeev Dvir and Allen Liu.
\newblock Fourier and circulant matrices are not rigid.
\newblock {\em Theory Of Computing}, 16(20):1--48, 2020.

\bibitem[DV78]{dubois1978new}
Eric Dubois and A~Venetsanopoulos.
\newblock A new algorithm for the radix-3 fft.
\newblock {\em IEEE Transactions on Acoustics, Speech, and Signal Processing},
  26(3):222--225, 1978.

\bibitem[FJ98]{frigo1998fftw}
Matteo Frigo and Steven~G Johnson.
\newblock Fftw: An adaptive software architecture for the fft.
\newblock In {\em Proceedings of the 1998 IEEE International Conference on
  Acoustics, Speech and Signal Processing, ICASSP'98 (Cat. No. 98CH36181)},
  volume~3, pages 1381--1384. IEEE, 1998.

\bibitem[HH11]{haynal2011generating}
Steve Haynal and Heidi Haynal.
\newblock Generating and searching families of fft algorithms.
\newblock {\em Journal on Satisfiability, Boolean Modeling and Computation},
  7(4):145--187, 2011.

\bibitem[JF06]{johnson2006modified}
Steven~G Johnson and Matteo Frigo.
\newblock A modified split-radix fft with fewer arithmetic operations.
\newblock {\em IEEE Transactions on Signal Processing}, 55(1):111--119, 2006.

\bibitem[Kiv21]{kivva2021improved}
Bohdan Kivva.
\newblock Improved upper bounds for the rigidity of kronecker products.
\newblock {\em arXiv preprint arXiv:2103.05631}, 2021.

\bibitem[L{\etalchar{+}}09]{lokam2009complexity}
Satyanarayana~V Lokam et~al.
\newblock Complexity lower bounds using linear algebra.
\newblock {\em Foundations and Trends{\textregistered} in Theoretical Computer
  Science}, 4(1--2):1--155, 2009.

\bibitem[LVB07]{lundy2007new}
T~Lundy and James Van~Buskirk.
\newblock A new matrix approach to real ffts and convolutions of length 2 k.
\newblock {\em Computing}, 80(1):23--45, 2007.

\bibitem[Mor73]{morgenstern1973note}
Jacques Morgenstern.
\newblock Note on a lower bound on the linear complexity of the fast fourier
  transform.
\newblock {\em Journal of the ACM (JACM)}, 20(2):305--306, 1973.

\bibitem[Pan86]{pan1986trade}
Victor~Ya Pan.
\newblock The trade-off between the additive complexity and the asynchronicity
  of linear and bilinear algorithms.
\newblock {\em Information processing letters}, 22(1):11--14, 1986.

\bibitem[Pap79]{papadimitriou1979optimality}
Christos~H Papadimitriou.
\newblock Optimality of the fast fourier transform.
\newblock {\em Journal of the ACM (JACM)}, 26(1):95--102, 1979.

\bibitem[Ram20]{ramya2020recent}
C~Ramya.
\newblock Recent progress on matrix rigidity--a survey.
\newblock {\em arXiv preprint arXiv:2009.09460}, 2020.

\bibitem[Ser17]{sergeev2017real}
Igor~Sergeevich Sergeev.
\newblock On the real complexity of a complex dft.
\newblock {\em Problems of Information Transmission}, 53(3):284--293, 2017.

\bibitem[SSK86]{suzuki1986new}
Yoiti Suzuki, Toshio Sone, and Kenuti Kido.
\newblock A new fft algorithm of radix 3, 6, and 12.
\newblock {\em IEEE transactions on acoustics, speech, and signal processing},
  34(2):380--383, 1986.

\bibitem[Val77]{valiant1977graph}
Leslie~G Valiant.
\newblock Graph-theoretic arguments in low-level complexity.
\newblock In {\em International Symposium on Mathematical Foundations of
  Computer Science}, pages 162--176. Springer, 1977.

\bibitem[VB04]{buskirk2004software}
James Van~Buskirk.
\newblock comp.dsp.
\newblock Usenet posts, January 2004.

\bibitem[Yav68]{yavne1968economical}
R~Yavne.
\newblock An economical method for calculating the discrete fourier transform.
\newblock In {\em Proceedings of the December 9-11, 1968, fall joint computer
  conference, part I}, pages 115--125, 1968.

\end{thebibliography}

\newpage

\appendix

\section{Exact Operation Count of Walsh-Hadamard Uprooted FFT} \label{appendix:section:hufft_runtime}

We now compute the operation count of Algorithm~\ref{alg:HUFFT}. We begin with the subroutine $TW$.

\begin{lemma} \label{lem:TWapproxcount}
    The system of recurrences
    \begin{align*}
        T_{TW}(N) &= 5N + T_{TW}(N/2) + 2T_{TWS}(N/4) \\
        T_{TWS}(N) &= 4N + T_{TWS2}(N/2) + 2T_{TWS}(N/4) \\
        T_{TWS2}(N) &= 5N + T_{TWS4}(N/2) + 2T_{TWS}(N/4) \\
        T_{TWS4}(N) &= 6N + T_{TWS2}(N/2) + 2T_{TWS}(N/4)
    \end{align*}
    solves to $$T_{TW}(N) = \frac{28}{9} N \log N + O(N).$$
\end{lemma}

\begin{proof}
    We first calculate
    \begin{align*}
        T_{TWS2}(N) &= 5N + (3N + T_{TWS2}(N/4) + 2T_{TWS}(N/8)) + 2T_{TWS}(N/4) \\
        &= 8N + T_{TWS2}(N/4) + 2T_{TWS}(N/8) + 2T_{TWS}(N/4) \\
        &= \frac{32}{3}N + 2\sum_{k = 0}^{\log N - 2}T_{TWS}(2^k), \\
        T_{TWS}(N) &= 4N + \left[ \frac{32}{3} \frac{N}{2} + 2 \sum_{k = 0}^{\log N - 3} T_{TWS}(2^k) \right] + 2T_{TWS}(N / 4) \\
        &= \frac{28}{3} N + 2 \sum_{k = 0}^{\log N - 2} T_{TWS}(2^k). \\
    \end{align*}
    
    We claim first that there is a constant $c \geq \frac{56}{9}$ such that $T_{TWS}(N) \leq \frac{28}{9} N \log N + c  N$ and prove it by strong induction. The base case follows by picking a sufficiently large $c$. For the inductive step, we have
    \begin{align*}
        T_{TWS}(N) &= \frac{28}{3} N + 2 \sum_{k = 0}^{\log N - 2} T_{TWS}(2^k) \\
        &\leq \frac{28}{3} N + 2\sum_{k = 0}^{\log N - 2} \frac{28}{9}(k)(2^k) + 2 c \sum_{k = 0}^{\log N - 2} 2^k \\
    \end{align*}
    
    Here we make use of the following two facts (that can be easily verified, for example, using computer algebra software):
    
    $$\sum_{k = 0}^{\log N - 2} k2^k = \tfrac{1}{2} (N \log N - 3N + 4),$$
    $$\sum_{k = 0}^{\log N - 2} 2^k = \frac{N}{2} - 1.$$
    
    Using these, we see
    
    \begin{align*}
        T_{TWS}(N) &\leq \frac{28}{3} N + \frac{28}{9}(N \log N - 3N + 4) + cN - 2c \\
        &= \frac{28}{9}N \log N + cN - 2c + \frac{112}{9}\\
        &\leq \frac{28}{9}N \log N + cN,
    \end{align*}
    where the last step follows if we pick $c \geq \frac{56}{9}$. This concludes the proof that $T_{TWS}(N) \leq \frac{28}{9}N \log N + O(N)$.
    
    Solving now for $T_{TW}(N)$:
        
    \begin{align*}
        T_{TW}(N) &= 5N + T_{TW}(N / 2) + 2T_{TWS}(N / 4) \\
        &= 5N + T_{TW}(N / 2) + 2 \left[ \frac{28}{9} \frac{N}{4} \log \frac{N}{4} + O(N) \right] \\
        &= \frac{14}{9}N \log N + T_{TW}(N / 2) + O(N)\\
        &= \frac{28}{9} N \log N + O(N).
    \end{align*}
    
    Where the full exact form, $T_{TW}(N) = \frac{28}{9} N \log N - \frac{112}{27} N - 2 \log N - \frac{2}{27}(-1)^{\log N} + 8$, is proven in Appendix \ref{appendix:section:hufft_runtime}.
\end{proof}

In fact, we can compute the count more precisely:

\begin{lemma} \label{lemma:hufft_twiddle_runtime}
    $T_{TW}(N) = \frac{28}{9} N \log N - \frac{112}{27} N - 2 \log N - \frac{2}{27}(-1)^{\log N} + 8$
\end{lemma}

\begin{proof}

    As we recall from the proof of Theorem \ref{theorem:hufft_correctness:hufft_runtime}, at each step of applying our ideas to manipulate the modified split-radix algorithm we only used numerical properties of the complex numbers and matrix multiplication to rearrange computations. In particular, if we use the folklore Fast Walsh-Hadamard Transform to compute $H'$, then Algorithm~\ref{alg:HUFFT} would have the same opeartion count as MSR. Indeed, our speedup comes from replacing the folklore Fast Walsh-Hadamard Transform with Algorithm 
    \ref{fast:hadamard:8x8}. This gives us the relationship
    
    $$T_{TW}(N) = T_{F}(N) - T_{H'_{folklore}}(N),$$
    where $T_F(N)$ is the operation count of the MSR algorithm (Algorithm~\ref{alg:msr}). We also define $H_{folklore}$ to be the folklore Fast Walsh-Hadamard Transform that uses $2N \log N$ real operations on a length $N$ complex input, and $H'_{folklore}$ to be Algorithm~\ref{hprime:alg}  if $H_{folklore}$ were used to compute the WHTs instead of Algorithm~\ref{fast:hadamard:8x8}.
    
    Using Lemma~\ref{hprime:complexity_vs_hadamard:lemma} and the fact that $T_{H_{folklore}} = 2N \log N$,
    \begin{align*}
        T_{H'_{naive}}(N) &= \sum_{j = 0}^{\log N} F(N, 2^j)T_{H_{naive}}(2^j) \\
        &= \sum_{j = 0}^{\log N} 2(2^j)i \\
        &= \frac{2}{3},N \lg N + \frac{4}{9}(-1)^{\log N} - \frac{4}{9}N
    \end{align*}
    which gives us
    \begin{align*}
        T_{TW}(N) &= T_{F}(N) - T_{H'_{naive}}(N) \\
        &= \frac{34}{9} N \log N - \frac{124}{27}N - 2 \log N + \frac{10}{27}(-1)^{\log N} + 8 - \left(\frac{2}{3}N \log N + \frac{4}{9}(-1)^{\log N} - \frac{4}{9}N\right) \\
        &= \frac{28}{9} N \log N - \frac{112}{27} N - 2 \log N - \frac{2}{27}(-1)^{\log N} + 8.
    \end{align*}
\end{proof}

\begin{corollary} \label{corollary:hufft_runtime_tot}
    $$T_{HUFFT}(N) \leq \frac{15}{4}N \log N - \frac{223}{108}N + o(N^{0.8})$$
\end{corollary}

\begin{proof}
    We simply add the results from Lemma~\ref{lemma:hufft_twiddle_runtime} and Theorem~\ref{hprime:complexity:thm} to get:
    \begin{align*}
        T_{HUFFT}(N) &= T_{H'}(N) + T_{TW}(N) \\
        &\leq \frac{23}{36}N \log N + \frac{25}{12}N + o(N^{0.8}) + \frac{28}{9} N \log N - \frac{112}{27} N - 2 \log N - \frac{2}{27}(-1)^{\log N} + 8 \\
        &= \frac{15}{4}N \log N - \frac{223}{108}N + o(N^{0.8}) - 2 \log N - \frac{2}{27}(-1)^{\log N} + 8 \\
        &= \frac{15}{4}N \log N - \frac{223}{108}N + o(N^{0.8}).
    \end{align*}
\end{proof}

\section{Properties of $H'$} \label{appendix:section:properties_hprime}

\begin{theorem} \label{hprime:partition:correctness}
    For an input $x$, the transformation $y = H'_Nx$ is equivalent to partitioning $x$ in a particular way and applying the WHT to each partition.
\end{theorem}

\begin{proof}
    We proceed by induction. For our base cases, we have $H'_1 = \begin{bmatrix} 1 \end{bmatrix}$ and $H'_2 = \begin{bmatrix} 1 & \\ & 1 \end{bmatrix}$. $H'_1$ is simply the same as the $1 \times 1$ WHT matrix $H_1$, while $H'_2$ is just two copies of the $1 \times 1$ WHT matrix along the diagonal, and thus the trivial way to compute $H'_2$ partitions a length 2 vector into two length 1 vectors and performs a size 1 WHT on each.
    
    Now suppose this statement is true for both $H'_{N/2}$ and $H'_{N/4}$. Recall that $H'_N$ has the structure
    
    $$H'_N = \unevenmatrix{H'_{N/2}}{H'_{N/4}}{H'_{N/4}}{H'_{N/4}}{-H'_{N/4}}.$$
    
    From the recursive structure of $H'_N$ since the upper left quadrant is an exact copy of $H'_{N/2}$ and the upper right quadrant is all zero, to compute $y[j]_{j = 0}^{N/2 - 1}$ (the first half of the output $y = H'_Nx$) we can just compute $H'_{N/2}x[j]_{j = 0}^{N/2 - 1}$. Thus, by induction to partition $x$ for $H'_{N}$ we first partition $x[j]_{j = 0}^{N/2 - 1}$ as if for $H'_{N/2}$.
    
    For the latter half of the output, again by the recursive structure of $H'_{N}$ we have
    \begin{align*}
        y[j]_{j = N/2}^{3N/4 - 1} = H'_{N/4}x[j]_{j = N/2}^{3N/4 - 1} + H'_{N/4}x[j]_{j = 3N/4}^{N - 1} \\
        y[j]_{j = 3N/4}^{N - 1} = H'_{N/4}x[j]_{j = N/2}^{3N/4 - 1} - H'_{N/4}x[j]_{j = 3N/4}^{N - 1} \\
    \end{align*}

    Since by induction we already assumed there is a way to partition each of $x[j]_{j = N/2}^{3N/4 - 1}$ and $x[j]_{j = 3N/4}^{N - 1}$ to compute $H'_{N/4}x[j]_{j = N/2}^{3N/4 - 1}$ and $H'_{N/4}x[j]_{j = 3N/4}^{N - 1}$ as WHTs on subsets of the inputs, let $s \subseteq x[j]_{j = N/2}^{3N/4 - 1}$ be any subset in this partition, and pick $s' \subseteq x[j]_{j = 3N/4}^{N - 1}$ to be the corresponding, identical subset of $x[j]_{j = 3N/4}^{N - 1}$ so that $x_j \in s$ if and only if $x_{j + N/4} \in s'$.
    
    In computing $y[j]_{j = N/2}^{3N/4 - 1}$ and $y[j]_{j = 3N/4}^{N - 1}$ we compute $H_{|s|}s + H_{|s|}s'$ and $H_{|s|}s - H_{|s|}s'$ \footnote{Here we use the fact that $|s| = |s'|$.} (WHTs on subsets of the input) which by the recursive structure of the WHT is equivalent to computing $H_{2|s|}(s \circ s')$, where $s \circ s'$ denotes the concatenation of $s$ and $s'$. Thus, to compute $y[j]_{j = N/2}^{3N/4 - 1}$ and $y[j]_{j = 3N/4}^{N - 1}$, which together make $y[j]_{j = N/2}^{N - 1}$, we can partition $x[j]_{j = N/2}^{N - 1}$ by taking every subset $s \subseteq x[j]_{j = N/2}^{3N/4 - 1}$ and corresponding $s' \subseteq x[j]_{j = 3N/4}^{N - 1}$, combine them into one subset $s \cup s'$, and perform the WHT $H_{2|s|}$ on it. Since $s$ and $s'$ came from partitions of $x[j]_{j = N/2}^{3N/4 - 1}$ and $x[j]_{j = 3N/4}^{N - 1}$, the result is a partition of $x[j]_{j = N/2}^{N - 1}$.
    
    Now that we have partitions of both $x[j]_{j = 0}^{N/2 - 1}$ and $x[j]_{j = N/2}^{N - 1}$ which when acted on by WHTs on each subset compute $y[j]_{j = 0}^{N/2 - 1}$ and $y[j]_{j = N/2}^{N - 1}$, together they give a partition of $x$ which when acted on by WHTs on each subset computes $y$.
\end{proof}

\begin{corollary}
    Furthermore, if the number of partitions of size $N_2$ for a length $N_1$ input is $F(N_1, N_2)$ then we have the recurrence $F(N_1, N_2) = F(N_1 / 2, N_2) + F(N_1 / 4, N_2 / 2)$.
\end{corollary}

\begin{proof}
    For the base cases, we can see $F(1, 1) = 1$, $F(2, 1) = 2$, $F(4, 1) = 2$ and $F(4, 2) = 1$.
    
    For any $N_1$ and $N_2$, we now compute $F(N_1, N_2)$. By the way our partition of an input of length $N_1$ is constructed, we exactly copy all of the size $N_2$ partitions of a length $N_1 /2$ input, giving $F(N_1 /2, N_2)$ copies. Meanwhile, we exactly copy and double the size of all partitions of a length $N_1/4$ input, so all size $N_2 / 2$ partitions of a length $N_1/4$ input, of which there are $F(N_1 / 4, N_2 / 2)$, turn into size $N_2 / 2$ partitions of a length $N_1$ input. Thus, $F(N_1, N_2) = F(N_1 / 2, N_2) + F(N_1 / 4, N_2 / 2)$.
\end{proof}

\section{Proof of Lemma \ref{lemma:sum-of-mod}} \label{appendix:proof:sum-of-mod}

\begin{proof}
    Define
    $$a_v(N) = \sum_{j = 0}^{\log N} F(N, 2^j) 2^j (1 \texttt{ if } j \equiv v \Mod 3 \texttt{, } 0 \texttt{ otherwise})$$
    
    Thus, 
    $$\sum_{j = 0}^{\log N}F(N, 2^j) \frac{2^j}{12}(j \Mod 3) = \frac{0a_0 + 1a_1 + 2a_2}{12},$$
    and we have the recurrence
    \begin{align*}
        a_0(N) &= \sum_{j = 0}^{{\log N}} F(N, 2^j) 2^j (1 \texttt{ if } j \equiv 0 \Mod 3 \texttt{, } 0 \texttt{ otherwise}) \\
        &= \sum_{j = 0}^{\log N} \left[ F(N/2, 2^j) + F(N/4, 2^{j - 1}) \right] 2^j (1 \texttt{ if } j \equiv 0 \Mod 3 \texttt{, } 0 \texttt{ otherwise}) \\
        &= \sum_{j = 0}^{\log N} F(N/2, 2^j) 2^j (1 \texttt{ if } j \equiv 0 \Mod 3 \texttt{, } 0 \texttt{ otherwise}) \\
        &+ \sum_{j = 0}^{\log N} F(N/4, 2^{j - 1}) 2^j (1 \texttt{ if } j \equiv 0 \Mod 3 \texttt{, } 0 \texttt{ otherwise}) \\
        &= a_0(N/2) + 2a_2(N/4),
    \end{align*}
    and likewise,
    \begin{align*}
        a_1(N) = a_1(N/2) + 2a_0(N/4), \\
        a_2(N) = a_2(N/2) + 2a_1(N/4).
    \end{align*}
    
    To solve this recurrence, we observe that
    
    $$\begin{bmatrix}
    a_0(N) \\ a_1(N) \\ a_2(N) \\ a_0(N/2) \\ a_1(N/2) \\ a_2(N/2) \end{bmatrix} = \begin{bmatrix} 1 & 0 & 0 & 0 & 0 & 2 \\ 0 & 1 & 0 & 2 & 0 & 0 \\ 0 & 0 & 1 & 0 & 2 & 0 \\ 1 & 0 & 0 & 0 & 0 & 0 \\ 0 & 1 & 0 & 0 & 0 & 0 \\ 0 & 0 & 1 & 0 & 0 & 0 \end{bmatrix} \begin{bmatrix}
    a_0(N/2) \\ a_1(N/2) \\ a_2(N/2) \\ a_0(N/4) \\ a_1(N/4) \\ a_2(N/4) \end{bmatrix}$$
    
    Call this $6 \times 6$ matrix $M$. It follows that for any $N$ which is a power of 2,
    
    $$\begin{bmatrix}
    a_0(N) \\ a_1(N) \\ a_2(N) \\ a_0(N/2) \\ a_1(N/2) \\ a_2(N/2) \end{bmatrix} = M^{\log N - 1} \begin{bmatrix}
    a_0(2) \\ a_1(2) \\ a_2(2) \\ a_0(1) \\ a_1(1) \\ a_2(1) \end{bmatrix} = M^{\log N - 1} \begin{bmatrix}
    2 \\ 0 \\ 0 \\ 1 \\ 0 \\ 0 \end{bmatrix}$$
    
    Furthermore, $M$ can be diagonalized into $M = SJS^{-1}$ for a diagonal matrix $J$, giving us
    
    $$\begin{bmatrix}
    a_0(N) \\ a_1(N) \\ a_2(N) \\ a_0(N/2) \\ a_1(N/2) \\ a_2(N/2) \end{bmatrix} = SJ^{\log N - 1}S^{-1} \begin{bmatrix}
    2 \\ 0 \\ 0 \\ 1 \\ 0 \\ 0 \end{bmatrix}.$$
    
    We now describe the matrices $J$ and $S$. 
    Define the complex numbers $a = \sqrt{-3 + 4i \sqrt{3}}$ and $b = \sqrt{-3 - 4i\sqrt{3}}$. The diagonal matrix $J$ is
    \begin{align*}
        J[0,0] &= -1, \\
        J[1,1] &= 2, \\
        J[2,2] &= \frac{1}{2} \left( 1 - a \right) \approx -0.2541 - 1.148i,\\
        J[3,3] &= \frac{1}{2} \left( 1 - b \right) \approx -0.2541 + 1.148i,\\
        J[4,4] &= \frac{1}{2} \left( 1 + b \right) \approx 1.2541 - 1.148i,\\
        J[5,5] &= \frac{1}{2} \left( 1 + a \right) \approx 1.2541 + 1.148i. \\
    \end{align*}
    
    The four complex entries all have magnitude smaller than $\sqrt{3}$, so
    \begin{align*}
        J^{\log N - 1}[0,0] &= \pm 1, \\
        J^{\log N - 1}[1,1] &= N/2, \\
        \left|J^{\log N - 1}[2,2]\right| &< \sqrt{3}^{\log N - 1} = O(N^{\log_2(3)/2}) = o(N^{0.8}),\\
        \left|J^{\log N - 1}[3,3]\right| &< \sqrt{3}^{\log N - 1} = O(N^{\log_2(3)/2}) = o(N^{0.8}),\\
        \left|J^{\log N - 1}[4,4]\right| &< \sqrt{3}^{\log N - 1} = O(N^{\log_2(3)/2}) = o(N^{0.8}),\\
        \left|J^{\log N - 1}[5,5]\right| &< \sqrt{3}^{\log N - 1} = O(N^{\log_2(3)/2}) = o(N^{0.8}).\\
    \end{align*}
    
    We next give the matrix $S$. Since it is quite large, we first write the first four columns:
    
    $$\begin{bmatrix}
    -1 & -2 & \tfrac{1}{2}(-1 + a + \tfrac{1}{4} (1 - a)^2 - \tfrac{1}{8}(1 - a)^3) & \tfrac{1}{2}(-1 + b + \tfrac{1}{4}(1 - b)^2 - \tfrac{1}{8}(1 - b)^3) \\
    -1 & -2 & \tfrac{1}{2}(-\tfrac{1}{4}(1 - a)^2 + \tfrac{1}{8}(1 - a)^3) & \tfrac{1}{2}(-\tfrac{1}{4}(1 - b)^2 + \tfrac{1}{8}(1 - b)^3) \\
    -1 & -2 & \tfrac{1}{2}(1 - a) & \tfrac{1}{2}(1 - b) \\
    1 & -1 & \tfrac{1}{2}(-2 + \tfrac{1}{2}(1 - a)  - \tfrac{1}{4}(1 - a)^2) & \tfrac{1}{2}(-2 + \tfrac{1}{2}(1 - b) - \tfrac{1}{4}(1 - b)^2) \\
    1 & -1 & \tfrac{1}{2}(\tfrac{1}{4}(1 - a)^2 + \tfrac{1}{2}(-1 + a)) & \tfrac{1}{2}(\tfrac{1}{4}(1 - b)^2 + \tfrac{1}{2}(-1 + b)) \\
    1 & -1 & 1 & 1
    \end{bmatrix}$$
    
    and then the last two columns:
    
    $$\begin{bmatrix} 
    \tfrac{1}{2}(-1 - b + \tfrac{1}{4}(1 + b)^2 - \tfrac{1}{8}(1 + b)^3) & \tfrac{1}{2}(-1 - a + \tfrac{1}{4}(1 + a)^2 - \tfrac{1}{8}(1 + a)^3) \\
    \tfrac{1}{2}(-\tfrac{1}{4}(1 + b)^2 + \tfrac{1}{8}(1 + b)^3) & \tfrac{1}{2}(-\tfrac{1}{4}(1 + a)^2 + \tfrac{1}{8}(1 + a)^3) \\
    \tfrac{1}{2}(1 + b) & \tfrac{1}{2}(1 + a) \\
    \tfrac{1}{2}(-2 + \tfrac{1}{2}(1 + b) - \tfrac{1}{4}(1 + b)^2) & \tfrac{1}{2}(-2 + \tfrac{1}{2}(1 + a) - \tfrac{1}{4}(1 + a)^2) \\
    \tfrac{1}{2}(\tfrac{1}{2}(-1 - b) + \tfrac{1}{4}(1 + b)^2) & \tfrac{1}{2}(\tfrac{1}{2}(-1 - a) + \tfrac{1}{4}(1 + a)^2) \\
    1 & 1
    \end{bmatrix}.$$
    
    To write out the full exact form of $S^{-1}$ would require at least several pages per entry of the matrix, and would not be particularly enlightening, so instead we write it with decimals rounded to three decimal places:
    
    $$\begin{bmatrix}
    -1/9 & -1/9 & -1/9 & 2/9 & 2/9 & 2/9 \\
    -1/9 & -1/9 & -1/9 & -1/9 & -1/9 & -1/9 \\
    -0.055 - 0.108i & 0.121 + 0.007i & -0.067 + 0.101i & -0.056 + 0.199i & -0.144 - 0.148i & 0.200 - 0.051i \\
    -0.055 + 0.108i & 0.121 - 0.007i & -0.067 - 0.101i & -0.056 - 0.199i & -0.144 + 0.148i & 0.200 + 0.051i \\
    0.055 - 0.108i & -0.121 + 0.007i & 0.067 + 0.101i & -0.111 - 0.090i & -0.023 + 0.141i & 0.133 - 0.051i \\
    0.0545 + 0.108i & -0.121 - 0.007i & 0.067 - 0.101i & -0.111 + 0.090i & -0.023 - 0.141i & 0.133 + 0.051i
    \end{bmatrix}.$$
    
    Since all but $J^{\log N - 1}[1,1]$ are $o(N^{0.8})$ they contribute $o(N^{0.8})$ to the overall operation count. Thus, we only care about the $1$-index row of $S^{-1}$ and the $1$-indexed column of $S$ (where the 1 indexed row and column are the second row and column counting from the top and left, respectively).
    
    The $1$-indexed row of $S^{-1}$ is $[-\tfrac{1}{9}, -\tfrac{1}{9}, -\tfrac{1}{9}, -\tfrac{1}{9}, -\tfrac{1}{9}, -\tfrac{1}{9}]$ and the $1$-indexed column of $S$ is $\begin{bmatrix} -2 \\ -2 \\ -2 \\ -1 \\ -1 \\ -1 \end{bmatrix}$. Multiplying it all out, we get $a_v(N) \leq N/3 + o(N^{0.8})$.
    
    Returning to the original statement of the lemma,
    \begin{align*}
        \sum_{j = 0}^{\log N}F(N, 2^j) \frac{2^j}{12}(j \Mod 3) &= \frac{0a_0 + 1a_1 + 2a_2}{12} \\
        &< N / 12 + o(N^{0.8})\\
        &< N/12 + o(N^{0.8}),
    \end{align*}
    as desired.
    
\end{proof}

\end{document}